\documentclass[aps,prl,reprint,superscriptaddress,amsmath,amssymb,longbibliography,floatfix]{revtex4-2}
\usepackage{amsmath,amssymb,amsthm}
\usepackage{mathtools}
\usepackage{graphicx,xcolor}
\usepackage{hyperref}
\newcommand{\Tr}{\operatorname{Tr}}
\newcommand{\norm}[1]{\left\lVert#1\right\rVert}
\newcommand{\ket}[1]{\lvert#1\rangle}
\newcommand{\bra}[1]{\langle#1\rvert}
\newcommand{\diam}{\operatorname{diam}}
\newtheorem{theorem}{Theorem}
\newtheorem{proposition}[theorem]{Proposition}
\newtheorem{lemma}[theorem]{Lemma}

\newsavebox{\statementcontent}
\newenvironment{statementbox}{%
 \par\smallskip\begin{lrbox}{\statementcontent}%
 \begin{minipage}{\columnwidth}\setlength{\topsep}{0pt}%
}{%
 \end{minipage}\end{lrbox}%
 \noindent\makebox[\columnwidth][c]{\fbox{\usebox{\statementcontent}}}%
 \par\smallskip
}
\hypersetup{pdftitle={Thermodynamic Stability of Classical Z2 Lattice Gauge Models},pdfauthor={Shunsuke Sengoku, Weiguang Cao, Haruki Watanabe},pdfsubject={General cell dimension at fixed spatial dimension; weak classical and quantum perturbations}}
\begin{document}
\title{Thermodynamic Stability of Classical \texorpdfstring{$\mathbb Z_2$}{Z2} Lattice Gauge Models}
\author{Shunsuke Sengoku}
\affiliation{Department of Physics, Hong Kong University of Science and Technology, Clear Water Bay, Hong Kong, China}

\author{Weiguang Cao}
\affiliation{Department of Physics, Hong Kong University of Science and Technology, Clear Water Bay, Hong Kong, China}
\affiliation{Center for Theoretical Condensed Matter Physics, Hong Kong University of Science and Technology, Clear Water Bay, Hong Kong, China}

\author{Haruki Watanabe}
\email{hwatanabe@ust.hk}
\affiliation{Department of Physics, Hong Kong University of Science and Technology, Clear Water Bay, Hong Kong, China}
\affiliation{Center for Theoretical Condensed Matter Physics, Hong Kong University of Science and Technology, Clear Water Bay, Hong Kong, China}
\affiliation{Institute for Advanced Study, Hong Kong University of Science and Technology, Clear Water Bay, Hong Kong, China}
\date{\today}
\begin{abstract}
Finite-temperature phases are traditionally studied through spontaneous symmetry breaking, but symmetry-breaking fields can remove their distinction, as in the Ising model. This motivates the search for stable nontrivial thermal phases without symmetry protection. For classical $\mathbb Z_2$ lattice gauge models with spins on $m$-dimensional cells in a fixed spatial dimension $d$, we prove joint analyticity of the thermodynamic free-energy density in temperature and perturbation strength at sufficiently low positive temperatures under sufficiently weak local perturbations, whenever $1\le m\le d-2$. Perturbations need not preserve internal or gauge symmetries and may include noncommuting quantum interactions. A separate high-temperature expansion yields an analytic region reaching infinite temperature at any fixed finite perturbation strength. These results establish thermodynamic stability and support the conjecture that low- and high-temperature regimes remain separated by a thermal transition without symmetry protection.
\end{abstract}
\maketitle

\emph{Introduction.---}
Spontaneous symmetry breaking is a central framework for finite-temperature phases and transitions~\cite{Stanley1971}. An external field that explicitly breaks the relevant symmetry can remove the phase distinction. For the ferromagnetic Ising model in $d\ge2$ dimensions, with exchange $J^{\rm I}>0$ and temperature $0<T<T_c^{\rm I}$ below its zero-field ordering transition, the thermodynamic free energy per site obeys
\begin{equation}
 f_{\rm I}(T,h)=f_{\rm I}(T,0)-m_{\rm sp}(T)|h|+o(|h|),
 \label{eq:ising-cusp}
\end{equation}
where $h$ is a longitudinal field and $m_{\rm sp}(T)>0$ is the spontaneous magnetization per site. The cusp gives a first-order transition at $h=0$: opposite infinitesimal fields select opposite ordered sectors. Although $f_{\rm I}(T,0)$ is analytic in $T$ throughout the ordered interval, it is not jointly analytic in $(T,h)$ near $h=0$~\cite{FriedliVelenik2017,Ott2020}.

At zero longitudinal field, the ordered and disordered regimes are separated by a thermal transition, whereas a nonzero field allows them to connect smoothly~\cite{LeeYang1952,FriedliVelenik2017,Ott2019}. This raises the question of whether a stable nontrivial phase can exist at finite temperature without symmetry protection, and remain separated from the high-temperature regime by a phase transition under arbitrary sufficiently weak local perturbations.

To address this question, we start from the plaquette-only 3D classical $\mathbb Z_2$ lattice gauge model~\cite{Wegner1971}, the toric code's magnetic sector~\cite{Kitaev2003}. Its field-free thermal states have been studied through topological entropy and classical memory~\cite{CastelnovoChamon2008,YangShiLee2026}. Weak-field robustness has been studied numerically and analytically~\cite{PoulinMelkoHastings2019,StahlEtAl2026}, gauge averaging addresses general classical perturbations~\cite{Thorngren2026}, and unperturbed low-temperature analyticity follows from cluster methods~\cite{ForsstromViklund2025}.

We place this example and the Ising model in one two-state cell family: spins occupy $m$-dimensional cells of a $d$-dimensional lattice, with one interaction per $(m+1)$-cell. The cases $m=0$ and $m=1$ are the Ising and plaquette models [Fig.~\ref{fig:cell-models}]. For $0\le m\le d-2$, the unperturbed low-temperature phase can be understood through spontaneous breaking of a $\mathbb Z_2$ $m$-form symmetry~\cite{Wegner1971,GaiottoEtAl2015}. We allow perturbations explicitly breaking this symmetry to ask whether stability survives without a microscopic symmetry distinction. The range $1\le m\le d-2$ combines two physical ingredients for stability without symmetry protection: locally indistinguishable flat sectors ($m\ge1$) and extended magnetic defects ($m\le d-2$).

Complementing classical memory-stability results, we prove joint free-energy analyticity in temperature and perturbation strength near zero perturbation at sufficiently low positive temperatures. This extends unperturbed analyticity to general local classical and noncommuting quantum perturbations throughout the geometric range above, without internal or gauge symmetry requirements. In contrast to the ordered Ising model, the free energy at sufficiently low positive temperature has no singularity at zero perturbation. We also establish a sufficient high-temperature analyticity domain. Together with the low-/high-temperature distinction under local quantum channels discussed at the end of this Letter, these results support a stable nontrivial thermal phase and motivate our conjecture of an intervening transition without symmetry protection. The detailed proofs are given in the Supplemental Material (SM)~\cite{SupplementalMaterial}.

Free-energy analyticity and local-channel phase distinction require different estimates.
Within the local-operation approach to mixed-state phases~\cite{CoserPerezGarcia2019,SangZouHsieh2024,MaKhemaniSang2026}, the companion paper~\cite{WatanabeFullPaper2026} constructs two-way local quantum channels and distinguishes low and high temperatures through a higher-cell extension of the winding invariant of Ref.~\cite{Watanabe2026}. Both analyses use gauge averaging, local fillings, weighted classical expansions, and an imaginary-time diagonal reduction. Their respective Supplemental Materials give the additional estimates needed for each result.

\emph{Spin models on $m$-cells.---}
A $j$-cell is a $j$-dimensional elementary cell of the cubic lattice: a vertex for $j=0$, a link for $j=1$, a plaquette for $j=2$, and a cube for $j=3$. On a periodic hypercubic lattice of integer spatial dimension $d\ge1$ and side $L$, choose $0\le m\le d-1$ and put one qubit on every $m$-cell $\ell$, with Pauli operators $Z_\ell$ and $X_\ell$, where $X_\ell$ flips the $Z_\ell$ eigenvalue. For each $(m+1)$-cell $p$, define
\begin{equation}
 B_p\coloneqq\prod_{\ell\in\partial p}Z_\ell,\qquad
 H_{d,m}\coloneqq-J\sum_pB_p,\qquad J>0.
 \label{eq:cell-model}
\end{equation}
For fixed $d,m$, we write $H_0\coloneqq H_{d,m}$. The boundary $\partial p$ consists of $2(m+1)$ physical $m$-cells. There are $N\coloneqq L^d$ unit cells, $\binom dmN$ qubits, and $\binom d{m+1}N$ interaction terms. Throughout, $d$ and $m$ are fixed as $L\to\infty$, and the lattice spacing and Boltzmann constant are one. We write $\beta\coloneqq1/T$ for inverse temperature.

Throughout this Letter, ``lattice gauge model'' is conventional terminology for a spin model on the full tensor-product Hilbert space. We impose no Gauss-law constraint and do not identify spin configurations related by gauge transformations. For $m\ge1$, $H_0$ has a local gauge symmetry, but perturbations need not preserve it.

\begin{figure}[t]
\centering
\includegraphics[width=216.33496pt,height=80.91190pt]{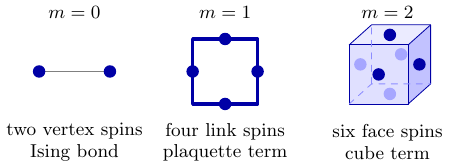}
\caption{Family of spin models defined on $m$-cells.
Each interaction multiplies the spins on the boundary of one cell:
two vertex spins for an Ising bond, four link spins for a plaquette,
and six face spins for a cube. Dots mark the centers of the spin-carrying
$m$-cells.
The same rule defines the model on $m$-cells in any spatial dimension $d\ge m+1$.
This model schematic is shared with Ref.~\cite{WatanabeFullPaper2026}.}
\label{fig:cell-models}
\end{figure}

For $m=0$, the boundary of a link is its two endpoints, so Eq.~\eqref{eq:cell-model} becomes the ferromagnetic Ising model:
\begin{equation}
 H_{d,0}=-J\sum_{\langle ij\rangle}Z_iZ_j.
 \label{eq:ising-model}
\end{equation}
Here $i,j$ label lattice vertices, and $\langle ij\rangle$ denotes nearest-neighbor pairs. No rescaling or additive constant is needed. For $m=1$, the same rule gives the four-link plaquette interaction
\begin{equation}
 H_{d,1}=-J\sum_p\prod_{\ell\in\partial p}Z_\ell.
 \label{eq:plaquette-model}
\end{equation}
This is the magnetic, or $Z$-plaquette, sector of the $d$-dimensional toric code, without electric star terms. The general family remains classical because all its $B_p$ commute and are diagonal in the product $Z$ basis. Figure~\ref{fig:cell-models} illustrates the interaction for vertex, link, and face spins.

Magnetic excitations are violations of the local interactions, whose signs obey an exact constraint. A classical configuration is specified by the spin signs $\sigma\coloneqq(\sigma_\ell)_\ell$, with $\sigma_\ell=\pm1$. The interaction value is $B_p(\sigma)\coloneqq\prod_{\ell\in\partial p}\sigma_\ell$. A configuration is flat when every $B_p(\sigma)=+1$; each violated interaction has $B_p(\sigma)=-1$ and costs $2J$.
For $m\le d-2$, every $(m+2)$-cell $q$, with boundary $(m+1)$-cells $\partial q$, satisfies the lattice Bianchi identity
\begin{equation}
 \prod_{p\in\partial q}B_p=I.
 \label{eq:bianchi}
\end{equation}
Here $I$ is the identity operator: each spin appears twice in the product and $Z_\ell^2=I$. Hence an even number of boundary interactions of each $q$ can be violated.

On the dual lattice, this constraint requires these excitations to form closed sets for $m\le d-2$. Represent each violated interaction on $p$ by the dual cell $p^\ast$ crossing it. This cell has dimension
\begin{equation}
 k\coloneqq d-m-1.
 \label{eq:defect-dimension}
\end{equation}
For $k\ge1$, each dual $(k-1)$-cell belongs to an even number of defect cells. A cell set's boundary consists of faces occurring an odd number of times, so these excitations have empty boundary and form closed $k$-dimensional sets: domain walls at $m=0$ and loops at $k=1$. At $m=d-1$, $k=0$ and defects are points; no $(m+2)$-cells impose this local constraint. Boundaries are taken in the appropriate physical or dual lattice. Cell sets add by symmetric difference: repeated cells cancel in pairs, and orientations are unnecessary for two-state spins.

This algebraic constraint remains exact under classical or quantum perturbations. Classical pyrochlore spin ice provides a useful comparison: its ice rule is energetic, and thermal monopoles can terminate defect strings, acting as dual ordering fields that round continuous transitions of the constrained model~\cite{WatanabeMotomeWatanabeSpinOne2026,WatanabeMotomeWatanabeSpinS2026}. The exact constraint excludes this endpoint mechanism in the present model. To prove stability, we also need local indistinguishability and bounds on the defect weights.

\begin{figure*}[t]
\centering
\includegraphics[width=537.42497pt,height=149.60681pt]{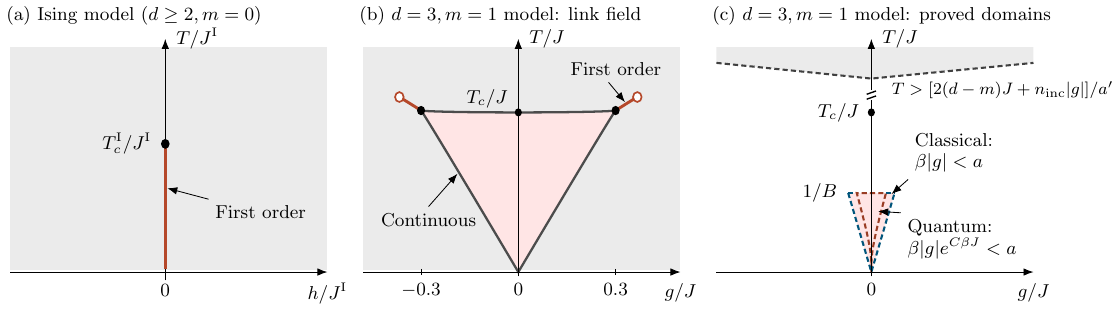}
\par\vspace{-5pt}
\caption{Schematic diagrams. (a) Ising model. (b) 3D plaquette model $(d=3,m=1)$ with a uniform link field.
In (a),(b), gray shading marks trivial regimes. Red shading in (b) marks the deconfined regime.
Solid red (dark gray) lines denote first-order (continuous) transitions; off-axis filled dots (open circles) denote multicritical points (critical endpoints).
(c) Low- (red) and high-temperature (gray) analytic domains of Theorems~\ref{thm:analyticity} and~\ref{thm:high-temperature}, restricted to $T>0$.
Dashed curves mark sufficient bounds, not phase boundaries. White regions are unclassified.}
\label{fig:domains}
\end{figure*}

\emph{Why both bounds on $m$ are needed.---}
The condition $k=d-m-1\ge1$ makes magnetic defects extended. A defect of $r$ connected dual $k$-cells costs $2Jr$, whereas the number of its possible shapes grows at most exponentially in $r$ for fixed $d,m$. At sufficiently low temperature, the energy penalty dominates this shape entropy. To control the coupling between defects and perturbations, we first work in a free box. This is a cubical region retaining every contained cell and only interaction terms supported wholly in the box, with all spins, including boundary spins, summed freely. In such a box, this positive defect dimension allows each component to be produced by a spin-flip pattern within a box whose diameter is proportional to its size.

The condition $m\ge1$ allows local gauge averaging while keeping flat sectors locally indistinguishable. The local $\mathbb Z_2$ gauge transformations are parametrized by signs $\eta_u$ on the $(m-1)$-cells and act as
\begin{equation}
 (\eta\sigma)_\ell\coloneqq\sigma_\ell\prod_{u\in\partial\ell}\eta_u.
 \label{eq:gauge-action}
\end{equation}
Here $\partial\ell$ consists of the $2m$ boundary $(m-1)$-cells of $\ell$. Every $B_p$, and hence $H_0$, is invariant. Averaging a perturbing Boltzmann weight over these local transformations removes terms with an uncanceled boundary. The surviving products form closed $m$-dimensional sets. In a free box these can be filled by $(m+1)$-cells, while on the torus, wrapping components need not bound.

Flat sectors are distinguished by Wilson products, namely spin products over closed $m$-dimensional surfaces wrapping around the torus. For $m\ge1$, a bounded contractible region cannot distinguish the uniform flat ensembles with different sector labels. At $m=0$, one spin distinguishes the all-up and all-down flat configurations, and no $(m-1)$-cell local transformation exists.

Consequently the stability theorem holds for
\begin{equation}
 m\ge1\;\text{and}\;d-m-1\ge1
 \;\Leftrightarrow\;1\le m\le d-2.
 \label{eq:stability-range}
\end{equation}
At $m=0$, a longitudinal field removes the Ising phase distinction [Eq.~\eqref{eq:ising-cusp} and Fig.~\ref{fig:domains}(a)]. At $m=d-1$, magnetic excitations are pointlike. The unperturbed thermodynamic free energy per unit cell $f$ is exactly
\begin{equation}
 f=-T[d\log2+\log\cosh(J/T)],
 \label{eq:point-defect-free-energy}
\end{equation}
which is analytic at every positive temperature. The Ising chain $(d,m)=(1,0)$ and the 2D plaquette model lie at this endpoint, whereas the 3D plaquette model lies in the stable range.

\emph{$m$-form symmetry.---}
Let $\Sigma$ be a closed set of dual $(d-m)$-cells and $\ell^\ast$ the dual cell crossing the physical $m$-cell $\ell$. The spin-flip operator
\begin{equation}
 U_\Sigma\coloneqq\prod_{\ell:\,\ell^\ast\in\Sigma}X_\ell
 \label{eq:form-symmetry}
\end{equation}
commutes with $H_0$: each $B_p$ meets an even number of flipped spins. Consider Wilson products on closed physical $m$-dimensional cell sets $C$,
\begin{equation}
 W_C\coloneqq\prod_{\ell\in C}Z_\ell.
 \label{eq:wilson-observable}
\end{equation}
Conjugation by $U_\Sigma$ changes the sign of $W_C$ exactly when $C$ and $\Sigma$ intersect an odd number of times. For $m\ge1$, if $\Sigma$ bounds a dual region, $U_\Sigma$ is a product of the local gauge transformations in Eq.~\eqref{eq:gauge-action}. For wrapping $\Sigma$, the operators $U_\Sigma$ generate the $m$-form symmetry, understood modulo those local transformations.

For $0\le m\le d-2$, the sufficiently low-temperature phase admits the interpretation of spontaneous breaking of the $m$-form symmetry generated by $U_\Sigma$ on wrapping $\Sigma$~\cite{Wegner1971,GaiottoEtAl2015}. For $m\ge1$, large contractible Wilson observables have a generalized perimeter law: the logarithm of their expectation scales with their $m$-dimensional support rather than the $(m+1)$-dimensional volume they enclose. Removing this local support-dependent decay is the usual higher-form order criterion. At $m=0$, this reduces to ordinary Ising order. For $H_0$, this symmetry breaking is defined in the thermodynamic limit, while its local gauge transformations remain unbroken~\cite{Elitzur1975}. The finite-volume canonical state remains symmetric, and charged winding observables have zero mean.

\emph{Perturbations and precise results.---}
We test stability under uniformly bounded local perturbations:
\begin{equation}
 H_g\coloneqq H_0+gV,\qquad V\coloneqq\sum_Xv_X.
 \label{eq:perturbed-model}
\end{equation}
Each dimensionless Hermitian $v_X$ has operator norm at most one and acts on a nonempty assigned set $X$ of at most $b$ physical $m$-cells and lattice diameter at most $R$. Each physical cell belongs to at most $n_{\mathrm{inc}}$ supports. These bounds are independent of $L$, and the operators are independent of temperature and $g$. Their fixed finite spatial period permits spatial averaging along compatible periodic sizes in the thermodynamic limit; it is not a protecting symmetry. Both $J$ and $g$ have units of energy.

At $\beta>0$, define the partition function and free energy per unit cell using the trace over all $2^{\binom dmN}$ states:
\begin{align}
 Z_L(\beta,g)&\coloneqq\Tr e^{-\beta H_g},\\
 f(\beta,g)&\coloneqq-\lim_{L\to\infty}\frac{\log Z_L(\beta,g)}{\beta N}.
 \label{eq:free-energy}
\end{align}

A classical perturbation is diagonal in the product $Z$ basis; quantum terms may be off-diagonal and need not commute with one another or with $H_0$. We require no invariance under either the local gauge transformations in Eq.~\eqref{eq:gauge-action} or the wrapping operators $U_\Sigma$. For example, $V\coloneqq-\sum_\ell Z_\ell$ explicitly breaks the wrapping transformations. The theorem therefore tests stability after explicitly breaking the microscopic symmetry that characterizes unperturbed order.

For $(d,m)=(3,1)$, this uniform link field gives the known phase diagram in Fig.~\ref{fig:domains}(b)~\cite{SomozaSernaNahum2021,BonatiPelissettoVicari2022}. Here $T_c$ is this 3D gauge model's zero-field transition temperature. The deconfined region survives a weak field, whereas the surrounding trivial regimes connect around the critical endpoints~\cite{FradkinShenker1979}. The following results extend weak-perturbation analyticity beyond this example.

\begin{statementbox}\begin{theorem}[Low-temperature stability]
\label{thm:analyticity}
For every fixed $d,m$ with $1\le m\le d-2$, there exist $B,a,C>0$, depending only on $d,m,b,R,n_{\mathrm{inc}}$, such that, for every allowed $V$, $f$ exists and is jointly real analytic in $(\beta,g)$ in the open domains
\begin{align}
 \beta J>B,\qquad\beta|g|<a
 &\quad\text{(classical)},\label{eq:classical-domain}\\
 \beta J>B,\qquad\beta|g|e^{C\beta J}<a
 &\quad\text{(quantum)}.\label{eq:quantum-domain}
\end{align}
\end{theorem}\end{statementbox}
At each sufficiently low positive temperature, Theorem~\ref{thm:analyticity} excludes an Ising-type cusp at zero perturbation. Joint analyticity gives a convergent power series in $(\beta,g)$ near every interior point, equivalently in $(T,g)$ for $T>0$. These sufficient domains do not cover the entire unperturbed low-temperature phase. The perturbation strength they permit vanishes as $T\to0$, so they imply no zero-temperature stability.

\begin{statementbox}\begin{theorem}[High-temperature analyticity]
\label{thm:high-temperature}
For every fixed $d,m$ with $0\le m\le d-1$, there exists $a'>0$, depending only on $m,b$, such that, for every allowed $V$, $f$ exists and is jointly real analytic in $(\beta,g)$ in the open domain
\begin{equation}
 \beta[2(d-m)J+n_{\mathrm{inc}}|g|]<a'.
 \label{eq:high-temperature-domain}
\end{equation}
This holds for both classical and quantum perturbations and needs no smallness of $|g|/J$.
\end{theorem}\end{statementbox}
Any fixed finite $g$ satisfies Eq.~\eqref{eq:high-temperature-domain} at sufficiently high temperature. Its explicit $d$ dependence counts the $2(d-m)$ interaction cells meeting each physical $m$-cell; $a'$ itself is independent of $d$. Figure~\ref{fig:domains}(c) shows these sufficient domains for $(d,m)=(3,1)$, with $2(d-m)J=4J$. For the uniform link field, the low-temperature domains lie well inside the deconfined region in (b); the intervening white region remains unclassified.

\emph{Classical proof strategy.---}
First, we average the zero-energy gauge variables while keeping the costly magnetic defects explicit. For classical $V$, taking the diagonal parts of the local terms leaves $V$ unchanged and preserves the locality bounds. We work in a free box so that closed perturbation chains admit local fillings. The superscript $\mathrm o$ denotes restriction to this box. Because $H_0^{\mathrm o}$ is gauge invariant, changing summation variables gives the exact identity
\begin{equation}
 Z_L^{\mathrm o}=\sum_\sigma e^{-\beta H_0^{\mathrm o}(\sigma)}
       \mathbb E_\eta e^{-\beta gV^{\mathrm o}(\eta\sigma)},
 \label{eq:gauge-average}
\end{equation}
where $\mathbb E_\eta$ uniformly averages the independent signs in Eq.~\eqref{eq:gauge-action}. The average acts on the Boltzmann weight, not on $V$~\cite{FoersterNielsenNinomiya1980,Testa1997,Thorngren2026}.

To express the gauge-averaged weight as a component sum, expand the diagonal interactions in sign products and combine equal nonempty supports. Write the resulting terms as $c_X\sigma_X$, where $c_X$ is the combined real coefficient and $\sigma_X\coloneqq\prod_{\ell\in X}\sigma_\ell$, and keep the scalar term separately. After extracting the $\cosh$ factors, expanding the remaining product assigns $\tanh(-\beta gc_X)$ to each selected term. Gauge averaging keeps precisely those selections in which every gauge sign occurs an even number of times.

To organize the surviving selections, join selected terms whose sets of uncanceled gauge signs overlap, after cancellations within each term. Separately, join flux cells through shared non-boundary dual $(k-1)$-faces. A perturbation component carries the product of its $\tanh$ factors, and a flux component with $r$ dual $k$-cells carries $e^{-2\beta Jr}$. Their coupling is a sign determined by crossings between the flux and a local filling of the closed perturbation chain. These connected components give an exact sum, with the scalar and $\cosh$ prefactors retained in the free energy.

To establish analyticity, expand the logarithm in connected clusters and bound their total absolute weight per unit cell uniformly in $L$~\cite{KoteckyPreiss1986,JacksonProcacciSokal2014,Ueltschi2004}. There are two counts to control: the number of component shapes and the number of relative positions that give a negative crossing sign. Component shapes grow exponentially in number, but the crossing sign increases the placement count only polynomially in the varying component's size and linearly in the fixed one's. Local fillings for both types control the placement count in both directions. Small $\beta|g|$ and $e^{-2\beta J}$ make the component weights decay fast enough to overcome both counts. Convergence is uniform on complex neighborhoods, proving analyticity of the exact free energy rather than only a formal series. A surface-order comparison of the original Hamiltonians identifies the open and periodic free energies at real parameters. The full construction, normalization, and convergence proof are in Sec.~I of the SM~\cite{SupplementalMaterial}.

\emph{Quantum reduction.---}
We use imaginary-time evolution to express the quantum partition function as a classical sum. For a product $Z$ eigenvector $|\sigma\rangle$, define the diagonal weight ratio
\begin{equation}
 \mathcal W_{\beta,g}(\sigma)\coloneqq
 e^{\beta H_0^{\mathrm o}(\sigma)}
 \langle\sigma|e^{-\beta H_g^{\mathrm o}}|\sigma\rangle.
 \label{eq:diagonal-weight-ratio}
\end{equation}
For each retained local term, its imaginary-time evolution is
\begin{equation}
 v_X(u)\coloneqq e^{u\beta H_0^{\mathrm o}}v_Xe^{-u\beta H_0^{\mathrm o}},
 \label{eq:imaginary-time-term}
\end{equation}
where $u\in[0,1]$ is dimensionless imaginary time. The exact interaction-picture identity is
\begin{equation}
 \mathcal W_{\beta,g}(\sigma)=
 \langle\sigma|\mathcal T\exp\left[-\beta g\int_0^1du\sum_Xv_X(u)\right]|\sigma\rangle,
 \label{eq:duhamel}
\end{equation}
Here $\mathcal T$ places larger $u$ to the left. Taking the trace gives an exact classical sum with effective perturbing energy $-\beta^{-1}\log\mathcal W_{\beta,g}(\sigma)$, which plays the role of the classical $gV^{\mathrm o}(\sigma)$.

The narrower sufficient quantum domain reflects imaginary-time norm growth. Since the $B_p$ commute, each evolved term stays within its enlarged support: $X$ together with the supports of all $B_p$ meeting $X$. Its norm grows by at most $e^{C\beta J}$ at real $\beta$. Join insertion labels with overlapping enlarged supports. Histories in disjoint groups factor in $|\sigma\rangle$. Expanding the ordered exponential and taking its connected logarithm gives a convergent local expansion of $\log\mathcal W$ for small $\beta|g|e^{C\beta J}$.

The resulting classical interaction has exponentially summable tails. Their decay absorbs the same polynomial placement factors as above, and the classical convergence argument applies with weighted bounds. The expansion retains all time orderings and scalar contributions, without truncating in $g$ or assuming a spectral gap. Section~II of the SM proves this reduction~\cite{SupplementalMaterial}; related quantum expansions appear in Refs.~\cite{BorgsKoteckyUeltschi1996,KoteckyUeltschi1999,FrohlichReyBelletUeltschi2001}. If $[H_0,V]=0$ on every sufficiently large compatible torus, Sec.~II D of the SM obtains the larger classical-type domain without the imaginary-time growth factor.

\emph{High-temperature expansion.---}
At high temperature, the connected expansion is controlled by inverse temperature times the interaction energy touching each spin, without using gauge structure. Write $H_g=\sum_Yh_Y$, combining equal supports. Each $h_Y$ acts on at most $\max\{2(m+1),b\}$ spins, and
\begin{equation}
 \sup_\ell\sum_{Y:\ell\in Y}\|h_Y\|
 \le 2(d-m)J+n_{\mathrm{inc}}|g|.
 \label{eq:local-energy}
\end{equation}
Expand $e^{-\beta H_g}$ about $\beta=0$. The normalized trace $2^{-\binom dmN}\Tr$ factors on disjoint supports, so the logarithm of the normalized partition function is a sum of connected clusters, even for noncommuting interactions. Small inverse temperature times this local energy bound ensures convergence and gives Eq.~\eqref{eq:high-temperature-domain} (Sec.~III of the SM)~\cite{SupplementalMaterial,HarrowMehrabanSoleimanifar2020}. The dimensionless free energy $\beta f$ is analytic through $\beta=0$, while $f$ contains the usual entropy term $-T\binom dm\log2$.

Both analyticity results extend to independent finite-dimensional onsite ancillary systems with bounded spectral widths and weak local classical or quantum couplings, assuming a common fixed spatial period for the joint-model data (Sec.~IV of the SM)~\cite{SupplementalMaterial}.

\emph{Scope and outlook.---}
Extended defects suppress large thermal fluctuations, while local gauge averaging turns arbitrary local symmetry-breaking weights into closed components with controlled defect couplings. Together these mechanisms give an analytic neighborhood of zero perturbation at every sufficiently low positive temperature for each fixed $d,m$ with $1\le m\le d-2$. This result covers the whole family; Fig.~\ref{fig:domains}(b) shows only the 3D uniform-link-field example.

Thermodynamic analyticity alone does not classify thermal states. The companion analysis~\cite{WatanabeFullPaper2026} distinguishes controlled low- and high-temperature regimes in this family through local quantum channels and nonlocal information. These state distinctions motivate, but do not prove, the following conjecture.

\begin{statementbox}
\noindent\textit{Conjecture (thermal phase separation).---}
For every fixed $d,m$ with $1\le m\le d-2$ and every allowed $V$, any continuous $(T,g)$ path with $T>0$ from the stable low-temperature domain to the high-temperature domain crosses a thermodynamic nonanalyticity of $f$ at finite temperature, even when both $T$ and $g$ vary.
\end{statementbox}
The conjecture is stronger than either analyticity theorem: our proof does not establish a singularity along every intervening path. Neither the theorems nor the conjecture determines the number, location, or order of transitions. These require further analysis for each cell model.

\label{last-core-page}
\begin{acknowledgments}
We thank Hal Tasaki for helpful discussions.
OpenAI Codex was used under the authors' direction to assist with drafting, organizing and checking derivations, and preparing schematic figures.
\end{acknowledgments}
\emph{Data availability.---}No datasets were generated or analyzed in this theoretical study.
%


\clearpage
\onecolumngrid
\setcounter{secnumdepth}{4}
\setcounter{section}{0}
\setcounter{subsection}{0}
\setcounter{subsubsection}{0}
\setcounter{equation}{0}
\setcounter{figure}{0}
\setcounter{table}{0}
\setcounter{theorem}{0}
\renewcommand{\theequation}{S\arabic{equation}}
\renewcommand{\thefigure}{S\arabic{figure}}
\renewcommand{\thetable}{S\arabic{table}}
\renewcommand{\thetheorem}{S\arabic{theorem}}
\providecommand{\theHsection}{}
\providecommand{\theHsubsection}{}
\providecommand{\theHsubsubsection}{}
\providecommand{\theHequation}{}
\providecommand{\theHfigure}{}
\providecommand{\theHtable}{}
\providecommand{\theHtheorem}{}
\renewcommand{\theHsection}{SM.\arabic{section}}
\renewcommand{\theHsubsection}{SM.\arabic{section}.\arabic{subsection}}
\renewcommand{\theHsubsubsection}{SM.\arabic{section}.\arabic{subsection}.\arabic{subsubsection}}
\renewcommand{\theHequation}{SM.\arabic{equation}}
\renewcommand{\theHfigure}{SM.\arabic{figure}}
\renewcommand{\theHtable}{SM.\arabic{table}}
\renewcommand{\theHtheorem}{SM.\arabic{theorem}}
\phantomsection
\pdfbookmark[0]{Supplemental Material}{supplement-start}
\begin{center}
{\large\bfseries Supplemental Material:\\[3pt]
Thermodynamic Stability of Classical $\mathbb Z_2$ Lattice Gauge Models\par}
\medskip
Shunsuke Sengoku,$^{1}$ Weiguang Cao,$^{1,2}$ and Haruki Watanabe$^{1,2,3}$\par
\smallskip
{\small\itshape
$^1$Department of Physics, Hong Kong University of Science and Technology, Clear Water Bay, Hong Kong, China\par
$^2$Center for Theoretical Condensed Matter Physics, Hong Kong University of Science and Technology, Clear Water Bay, Hong Kong, China\par
$^3$Institute for Advanced Study, Hong Kong University of Science and Technology, Clear Water Bay, Hong Kong, China\par}
\smallskip
\today
\end{center}
\medskip

\noindent This Supplemental Material contains complete proofs of the
low- and high-temperature analyticity results and can be read independently
of the Letter. Section~\ref{sm:cl:section} defines the model and locality
assumptions, derives the exact gauge-averaged partition function, proves
cluster convergence and its analytic bulk limit, and returns from an
auxiliary open box to periodic boundaries. Section~\ref{sm:sec:quantum}
derives the local imaginary-time representation for noncommuting
perturbations, applies the weighted classical expansion, and gives the
stronger analyticity domain when the total perturbation commutes with $H_0$.
Section~\ref{sm:ht:section} proves the high-temperature result by an exact
normalized-trace expansion. Section~\ref{sm:anc:section} extends both
analyticity results to independent onsite ancillary systems with weak local
couplings, retaining their full free-energy contribution.

Throughout, $d,m$ are fixed. The model family
has $0\le m\le d-1$, the low-temperature proof requires
$1\le m\le d-2$, and the high-temperature proof applies to the full family.
Physical $m$-cells are denoted by $\ell$, interaction $(m+1)$-cells by $p$,
and $H_0\coloneqq H_{d,m}$. Every partition function uses the full spin
Hilbert space, with no Gauss-law projection. Perturbations need preserve
neither local gauge transformations nor the wrapping $m$-form symmetry.
The notation agrees with the Letter, and all definitions and proof steps
needed below are given here.
\section{Classical proof}
\label{sm:cl:section}\label{sm:sec:pressure}

\subsection{Model, conventions, and low-temperature statement}
\label{sm:cl:model}
Fix an integer spatial dimension $d\ge1$ and cell dimension $m$ with
$0\le m\le d-1$. The low-temperature theorem below uses the subrange
$1\le m\le d-2$; the high-temperature theorem applies to the full family.
A $j$-cell is a $j$-dimensional elementary cell of
the cubic lattice: a vertex for $j=0$, a link for $j=1$, a plaquette for
$j=2$, and a cube for $j=3$. On the periodic lattice of side
$L$ there are $N\coloneqq L^d$ unit cells and $\binom djN$ cells of
dimension $j$. Place a qubit on every $m$-cell $\ell$, with Pauli
operators $Z_\ell$ and $X_\ell$. For each
$(m+1)$-cell $p$, its boundary $\partial p$ contains $2(m+1)$
$m$-cells. The generalized gauge Hamiltonian~\cite{SM-Wegner1971} is
\begin{equation}
 B_p\coloneqq\prod_{\ell\in\partial p}Z_\ell,\qquad
 H_{d,m}\coloneqq-J\sum_pB_p.\label{sm:cl:model-plaquettes}
\end{equation}
For fixed $d,m$, write $H_0\coloneqq H_{d,m}$, and introduce the perturbation by
\begin{equation}
 V\coloneqq\sum_Xv_X,\qquad
 H_g\coloneqq H_0+gV.\label{sm:cl:model-hamiltonian}
\end{equation}
The family has
$\binom dmN$ qubits and $\binom d{m+1}N$ interaction terms.
At $m=0$ this is the nearest-neighbor ferromagnetic Ising model, with no
rescaling or additive constant. At $m=1$ it is the four-link plaquette
model, the magnetic sector of the toric code; no electric star term is
included. All unperturbed interactions commute and are diagonal in the
product $Z$ basis.

Here $J>0$ and $g$ have units of energy. Every dimensionless Hermitian
$v_X$ has operator norm at most one and assigned nonempty support $X$ of at most
$b$ physical $m$-cells and lattice diameter at most $R$. Each physical
cell belongs to at most $n_{\mathrm{inc}}$ assigned supports. The bounds $b,R,n_{\mathrm{inc}}$,
the operators, and their finite spatial period are fixed independently of
$L,\beta,g$. Take $L$ through compatible periodic sizes. These are
locality and thermodynamic-limit assumptions. The lattice spacing and
Boltzmann constant are one.
Constants in estimates may change from line to line and depend on
$d,m,b,R,n_{\mathrm{inc}}$ and explicitly chosen exponential reserves, never on volume.

At $\beta\coloneqq1/T>0$, define
\begin{align}
 Z_L(\beta,g)&\coloneqq\Tr e^{-\beta H_g},\label{sm:cl:model-partition}\\
 f(\beta,g)&\coloneqq-\lim_{L\to\infty}\frac{\log Z_L(\beta,g)}{\beta N}.
 \label{sm:cl:model-free-energy}
\end{align}
We use ``gauge model'' or ``gauge theory'' as conventional terminology for
this unconstrained spin model. The trace runs over the full
$2^{\binom dmN}$-dimensional spin Hilbert space, with no Gauss-law
projection, no identification of spin configurations related by local
gauge transformations, and no restriction to a winding sector. The proof
uses the local gauge symmetry of $H_0$ for $m\ge1$, but imposes no
local gauge or internal symmetry on $V$, including no invariance under
the wrapping $m$-form transformations defined below. All free energies are per unit cell,
not per qubit.

The term ``classical'' means that the total perturbation $V$ is diagonal
in the product $Z$ basis. We may then choose diagonal local terms:
replacing each original $v_X$ by its diagonal part leaves their sum equal
to $V$ and does not increase their norms or supports. The quantum proof
retains the original operators.

\begin{theorem}[Low-temperature analyticity]
\label{sm:thm:low-temperature}
For every fixed $d,m$ with $1\le m\le d-2$, there exist $B,a,C>0$,
depending only on $d,m,b,R,n_{\mathrm{inc}}$, such that, for every allowed $V$,
$f$ exists and is jointly real analytic in $(\beta,g)$ in the open domains
\begin{align}
 \beta J>B,\qquad\beta|g|<a
 &\quad\text{(classical)},\label{sm:cl:classical-domain}\\
 \beta J>B,\qquad\beta|g|e^{C\beta J}<a
 &\quad\text{(quantum)}.\label{sm:cl:quantum-domain}
\end{align}
\end{theorem}
Analyticity is equivalent in $(T,g)$ for $T>0$. This section proves the
classical statement, and Sec.~\ref{sm:sec:quantum} proves the quantum
statement. The inequalities specify sufficient domains of analyticity
and do not locate phase boundaries.
The high-temperature theorem in Sec.~\ref{sm:ht:section} also applies
throughout $0\le m\le d-1$. These sufficient low-temperature domains
do not include fixed nonzero $g$ as $T\to0$ and imply no zero-temperature
stability statement.

To establish real analyticity after the thermodynamic limit, we construct
a locally uniform holomorphic limit on complex neighborhoods of the real
parameters. Three separate steps are necessary: identify the exact
finite-volume partition function, prove convergence of its connected
logarithm, and pass to the bulk limit on a common complex neighborhood.

\subsection{Symmetry and scope of the low-temperature proof}
\label{sm:cl:symmetry-scope}

For $m\ge1$, the proof uses the local gauge invariance of $H_0$
under Eq.~\eqref{sm:cl:gauge-action}. On the torus, a closed set $\Sigma$
of dual $(d-m)$-cells also defines
$U_\Sigma\coloneqq\prod_{\ell:\ell^\ast\in\Sigma}X_\ell$,
where $\ell^\ast$ is the dual cell crossing $\ell$.
Closed means that every boundary face occurs an even number of times.
If $\Sigma$ bounds
a dual region, this is a product of local gauge transformations.
Wrapping $\Sigma$ generate the $m$-form symmetry modulo those local
transformations~\cite{SM-GaiottoEtAl2015}. The symmetry-breaking interpretation in the Letter
concerns these wrapping operators. The local gauge transformations
remain unbroken~\cite{SM-Elitzur1975}. No invariance of $V$ under either kind is required.

The two restrictions have different roles: $m\ge1$ permits local gauge
averaging and locally indistinguishable flat sectors, whereas
$k\coloneqq d-m-1\ge1$ gives extended magnetic defects with local
fillings. At the endpoint $m=d-1$, expanding each factor as
$e^{\beta J B_p}=\cosh(\beta J)[1+B_p\tanh(\beta J)]$ leaves only
the empty and full sets of $d$-cells after the spin sum: these are the only
sets with empty boundary on the connected torus. Hence
\begin{equation}
 Z_L(\beta,0)=2^{dN}[\cosh(\beta J)]^N
                 [1+\tanh(\beta J)^N].
 \label{sm:cl:point-defect-partition}
\end{equation}
Taking its logarithm per unit cell gives the endpoint free energy quoted
in the Letter. From the next subsection through the quantum proof we assume
$1\le m\le d-2$. The high-temperature proof needs neither restriction.

\subsection{The free box, local fillings, and exact flux multiplicity}
\label{sm:cl:open-box}
We first work in a free box, where closed physical chains have no winding
obstruction. The local filling lemma below treats both interior fluxes
and fluxes that reach the boundary. We then count the spin configurations
realizing each allowed flux. Section~\ref{sm:cl:boundary-comparison}
identifies the resulting bulk free energy with that of the periodic system. The local filling and defect geometry used below are
shared with the companion channel construction~\cite{SM-WatanabeFullPaper2026}.
The companion uses local fillings to implement defect removal and insertion.

Use the cubical box $\Lambda_L$ with vertices
$\{0,\ldots,L-1\}^d$, retaining every contained cell and summing all
physical $m$-cell signs freely, including on the boundary. Its cell counts are
\begin{equation}
 N_j^{\mathrm o}\coloneqq\binom dj(L-1)^jL^{d-j},\qquad
 N\coloneqq L^d,\qquad N_p\coloneqq N_{m+1}^{\mathrm o}.
 \label{sm:cl:open-counts}
\end{equation}
Retain precisely those original perturbing terms whose supports lie in
this box, giving $V^{\mathrm o}$. Write
\begin{equation}
 H_0^{\mathrm o}\coloneqq-J\sum_{p\subset\Lambda_L}B_p,\quad
 H_g^{\mathrm o}\coloneqq H_0^{\mathrm o}+gV^{\mathrm o},\quad
 Z_L^{\mathrm o}\coloneqq\Tr e^{-\beta H_g^{\mathrm o}}.
 \label{sm:cl:open-hamiltonian}
\end{equation}

A configuration is $\sigma\coloneqq(\sigma_\ell)_\ell$, with $\sigma_\ell=\pm1$;
for an $m$-cell set $Y$, put $\sigma_Y\coloneqq\prod_{\ell\in Y}\sigma_\ell$.
An empty product is one. A set of cells is also called a chain. Addition
of chains always means symmetric difference: repeated cells cancel in
pairs. Its boundary consists of faces present an odd number of times.
Thus ``modulo two'' below requires only parity, not orientations.
The empty chain is denoted by $\varnothing$ or $0$.

Independent signs $\eta_u$ on the $(m-1)$-cells $u$ (vertices when $m=1$) act by
\begin{equation}
 (\eta\sigma)_\ell\coloneqq\sigma_\ell\prod_{u\in\partial\ell}\eta_u.
 \label{sm:cl:gauge-action}
\end{equation}
Here $\partial\ell$ is the set of $2m$ boundary $(m-1)$-cells of $\ell$.
Each $(m-1)$-cell occurs twice in the boundary of a boundary, so every
$B_p$ is unchanged. Each choice of $\eta$ permutes all configurations.
Different choices can implement the same transformation. This does not
affect the averaging identity: each choice separately gives a bijective
change of summation variables, so we may average uniformly over all
$2^{N_{m-1}^{\mathrm o}}$ choices.
This local transformation exists because $m\ge1$.

A negative interaction $B_p=-1$ is represented by the dual cell crossing
$p$. Its dimension is
\begin{equation}
 k\coloneqq d-m-1\ge1.
\end{equation}
The dual cellulation is clipped at the physical boundary. The relative
boundary $\partial_{\rm rel}$ is the ordinary boundary with all cells lying
wholly on the physical boundary omitted; these boundary cells are auxiliary.
The occupied dual $k$-cells form the flux $E(\sigma)$. For every retained
$(m+2)$-cell $t$,
\begin{equation}
 \prod_{p\in\partial t}B_p=1.
 \label{sm:cl:cube-rule}
\end{equation}
Equivalently, every non-boundary dual $(k-1)$-face meets an even number
of occupied $k$-cells. There is no parity constraint on a face lying on
the physical boundary.

Flux components are connected through shared non-boundary $(k-1)$-faces.
Distinct components cannot cancel a boundary at such a face, so each
component satisfies the parity rule separately. For $k=1$ this
is the familiar even-degree graph with paths allowed to end at the
boundary. Let $\mathfrak P_L^{\mathrm o}$ be all sets satisfying this rule.

The lemma gives two constructions used later in the partition sum.
A physical filling rewrites a closed spin product as a product of
interaction signs $B_p$. A dual filling specifies which spins to flip
to create a prescribed flux.
For link spins in three dimensions, sweep a closed physical loop across
plaquettes until it disappears [Fig.~\ref{sm:fig:physical-sweeps}]. The
recorded plaquettes form a surface bounded by the loop. To create a dual
flux loop, fill it with a dual surface and flip the physical links
crossing that surface [Fig.~\ref{sm:fig:boundary-flux-realization}].
Starting from all spins $+1$, only plaquettes crossed by the prescribed
flux become negative.

A flux path ending at the open boundary is first
closed by an auxiliary path on that boundary, which adds no physical
flux. If the padded box containing the flux meets opposite physical
faces, we use the whole boundary for this closure. Then $L=O(r)$, where
$r$ counts the occupied flux cells, so the larger filling still fits in
a box of diameter $O(r)$ [Fig.~\ref{sm:fig:boundary-geometric-cases}].
The lemma bounds filling diameters, not the number of filling cells.

\begin{lemma}[Local fillings, including boundary components]
\label{sm:lem:open-fillings}
\textup{(i) Physical fillings.}
For fixed $c_0>0$ and integer scale $j\ge1$, a physical $m$-chain $C$ with
$\partial C=\varnothing$, contained in a box of diameter at most $c_0j$,
bounds an $(m+1)$-chain $S_C$ in a box of diameter at most $C_1j$ inside
$\Lambda_L$:
\begin{equation}
 \partial S_C=C.
 \label{sm:cl:physical-filling}
\end{equation}
\textup{(ii) Flux realization.}
A connected allowed dual flux component $P_{\rm f}$ of $r\ge1$ occupied
$k$-cells is realized by a physical sign pattern $F_{P_{\rm f}}$ whose
negative $m$-cells lie in a box of diameter at most $C_2r$:
\begin{equation}
 B_p(F_{P_{\rm f}})=
 \begin{cases}-1,&p\text{ is crossed by }P_{\rm f},\\
 +1,&\text{otherwise}.
 \end{cases}
 \label{sm:cl:open-fillings}
\end{equation}
The constants are independent of $L,j,r$.
\end{lemma}

Here $F_{P_{\rm f}}$ is a complete assignment of signs to physical
$m$-cells. Its value $(F_{P_{\rm f}})_\ell=-1$ means that cell $\ell$
is flipped relative to the all-positive configuration.
The quantity $B_p(F_{P_{\rm f}})$ is the product of these signs around
$p$. Part (ii) therefore prescribes the negative interactions produced
by this spin-flip pattern.

\begin{figure}[htbp]
 \centering

 \includegraphics[width=510.00719pt,height=194.16084pt]{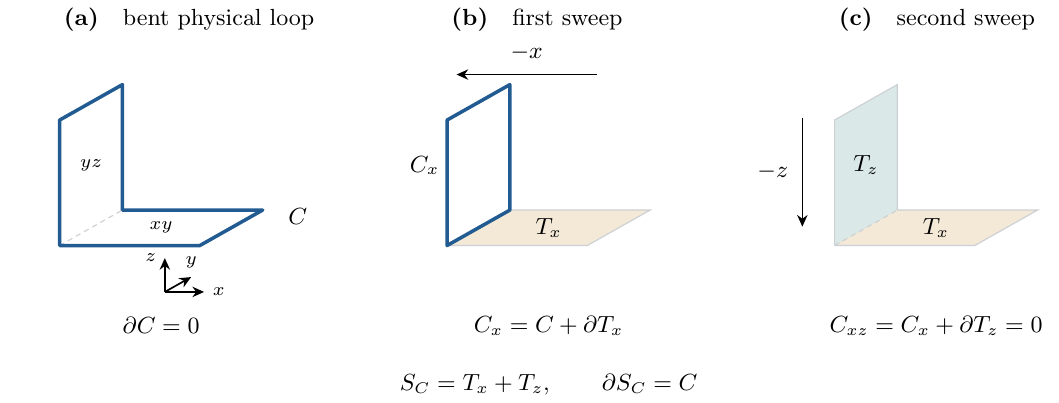}
 \caption{\textbf{Physical filling by coordinate sweeps} ($d=3$, $m=1$).
 Blue lines show the remaining chain; $T_x,T_z$ are the recorded
 plaquette sets, shown by shading.
 (a)~A bent six-link loop $C$.
 (b)~The $-x$ sweep records $T_x$, leaving $C_x=C+\partial T_x$.
 (c)~The $-z$ sweep records $T_z$ and removes $C_x$, so
 $\partial(T_x+T_z)=C$. Gray lines indicate lattice edges;
 all additions are modulo two.}
 \label{sm:fig:physical-sweeps}
\end{figure}

\begin{proof}
\textup{(i) Physical fillings: sweep and record.}
In a rectangular cubical box, sweep a chain in one coordinate to the
lowest face, recording the next-dimensional cells crossed. The boundary
of the recorded chain consists of the original chain, its projection
onto that face, and the pieces swept out by its boundary. For a closed
chain the last pieces cancel modulo two. Write $C^{(0)}\coloneqq C$, let
$C^{(i)}$ be its successive projection after the first $i$ coordinate sweeps,
and let $T_i$ be the $(m+1)$-chain recorded in the $i$th sweep. Then
\begin{equation}
 \partial T_i=C^{(i-1)}+C^{(i)},\qquad i=1,\ldots,d.
 \label{sm:cl:sweep-step}
\end{equation}
After all coordinates, the image lies at a single vertex and hence
$C^{(d)}=\varnothing$ because $m\ge1$. Adding the recorded chains therefore gives
\begin{align}
 S_C&\coloneqq T_1+\cdots+T_d,\label{sm:cl:sweep-filling}\\
 \partial S_C&=(C^{(0)}+C^{(1)})+\cdots+(C^{(d-1)}+C^{(d)})=C.
 \label{sm:cl:sweep-cancellation}
\end{align}
All intermediate projections cancel in pairs. Figure~\ref{sm:fig:physical-sweeps}
uses the coordinate order $x,z,y$: the first two sweeps record
plaquettes and the last records none.
Every recorded cell stays inside the original coordinate box, proving
(i). The same construction fills a closed chain of any dimension
$1,\ldots,d-1$.

\smallskip
\textup{(ii) Flux realization: close at the boundary, then fill.}
Connectedness of $P_{\rm f}$ gives a containing coordinate box of diameter
$O(r)$: adjoining one cell across a shared face increases the distance
from the first cell by at most a fixed amount. Pad this box by a fixed
number of lattice spacings and clip it at the physical boundary, obtaining
$D$. Define $\partial_{\rm phys}D\coloneqq D\cap\partial\Lambda_L$;
artificial faces of $D$ are not included. Since
$\partial_{\rm rel}P_{\rm f}=0$, the ordinary boundary
$\partial P_{\rm f}$ lies in $\partial_{\rm phys}D$. We will construct an
auxiliary dual $k$-chain $Q$ there satisfying
\begin{equation}
 Q\subset\partial_{\rm phys}D,\qquad \partial Q=\partial P_{\rm f}.
 \label{sm:cl:boundary-completion}
\end{equation}
For an already closed flux, take $Q=\varnothing$.

\smallskip
\textup{Case 1: no opposite physical faces.}
If $D$ meets no physical face, $P_{\rm f}$ is already closed. Otherwise
$\partial_{\rm phys}D$ is a union of adjacent faces with a common
intersection. Coordinate sweeps within these faces contract that union
to a vertex [Fig.~\ref{sm:fig:boundary-geometric-cases}(a)]. For $k>1$,
$\partial P_{\rm f}$ is a closed positive-dimensional chain and these
sweeps supply $Q$. For $k=1$, its endpoints occur in even number because they are the
boundary of an edge set. Pair them by paths in the connected physical
boundary part. The union modulo two of those paths supplies $Q$.

\smallskip
\textup{Case 2: opposite physical faces.}
If the padded box meets opposite physical faces, its $O(r)$ diameter
implies $L=O(r)$. Enlarge $D$ to the full box, so
$\partial_{\rm phys}D=\partial D$ without losing the $O(r)$ diameter
bound. This avoids disconnected boundary patches of the smaller box
[Fig.~\ref{sm:fig:boundary-geometric-cases}(b,c)].

The chain $\partial P_{\rm f}$ fills within $\partial D$ for $1\le k<d$.
For $k=1$, pair its even number of endpoints along the connected boundary.
For $k>1$, split it into cells on one chosen face and the remaining cells.
The first part has its boundary on the rim of that face. By the same
argument one dimension lower, with even endpoint pairing as the base
case, this boundary fills on the rim. Joining this rim filling to the
first part gives a closed chain on the face, which a coordinate sweep
fills there. Adding the boundary of that filling moves the original
chain off the interior of the chosen face. The remaining faces can be
swept onto the opposite face and then to a vertex, producing the rest
of $Q$. This induction uses $k-1<d-1$ and therefore still works at $k=d-1$.
The theorem's stronger restriction $k\le d-2$ comes separately from
$m\ge1$, which permits local gauge averaging and makes the physical
chains in (i) positive-dimensional.

\smallskip
\textup{Complete the flux realization.}
In both cases, Eq.~\eqref{sm:cl:boundary-completion} gives
\begin{equation}
 \partial(P_{\rm f}+Q)=\partial P_{\rm f}+\partial Q=\varnothing.
 \label{sm:cl:completed-flux}
\end{equation}
The coordinate sweep now supplies a dual $(k+1)$-chain $T$ in $D$ with
\begin{align}
 \partial T&=P_{\rm f}+Q,\label{sm:cl:ordinary-filling}\\
 \partial_{\rm rel}T&=P_{\rm f}.
 \label{sm:cl:relative-filling}
\end{align}
The second equality omits $Q$, whose cells lie on the physical boundary.
Discard filling cells wholly on that boundary and flip the physical
$m$-cells crossed by the remaining cells of $T$. A single flipped cell
$\ell$ changes precisely the interactions $p$ containing it, whose dual
cells form $\partial_{\rm rel}\ell^\ast$. Multiple flips cancel in pairs,
so the resulting negative interactions are exactly those dual to
$\partial_{\rm rel}T=P_{\rm f}$
[Fig.~\ref{sm:fig:boundary-flux-realization}]. Auxiliary boundary cells
add no physical flux and never enter connectivity or cost. All flipped
cells remain within $D$ up to fixed padding, proving (ii).
\end{proof}

\begin{figure}[!hbp]
 \centering

 \includegraphics[width=510.00967pt,height=328.27411pt]{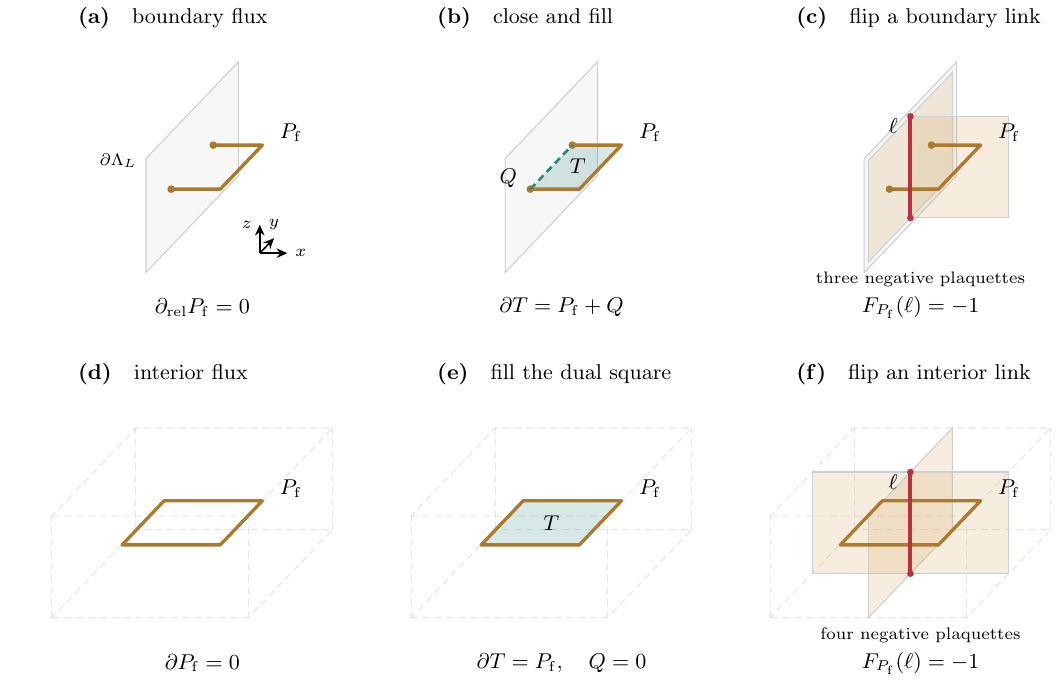}
 \caption{\textbf{Realizing a dual flux} ($d=3$, $m=1$).
 Orange lines show $P_{\rm f}$, dashed green lines the auxiliary $Q$,
 teal shading the filling $T$, and red lines the flipped physical link.
 (a--c)~Close a boundary path with $Q$, fill it, and flip the link
 dual to $T$, creating three negative plaquettes.
 (d--f)~An interior square needs no $Q$; the dual link flip creates
 four negative plaquettes. Pale orange faces mark these plaquettes;
 all other spins and plaquettes are positive. Dashed gray frames
 indicate a local cell patch, not the physical boundary.}
 \label{sm:fig:boundary-flux-realization}
\end{figure}
\begin{figure}[htbp]
 \centering

 \includegraphics[width=509.99582pt,height=196.28840pt]{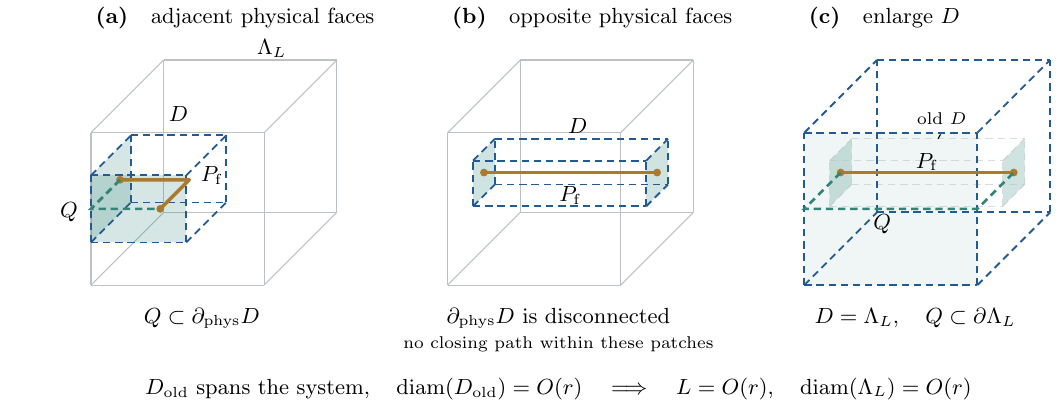}
 \caption{\textbf{Boundary completion in the two geometric cases.}
 Blue dashed edges bound $D$, and teal marks its physical boundary portions.
 Orange lines show $P_{\rm f}$, and dashed green lines show $Q$.
 (a)~Adjacent faces allow completion within $\partial_{\rm phys}D=D\cap\partial\Lambda_L$.
 (b)~Opposite faces can give disconnected patches, preventing completion
 within $\partial_{\rm phys}D$. Artificial faces lie in the bulk and cannot be used.
 (c)~Enlarging $D$ to $\Lambda_L$ provides a boundary route for $Q$.
 The diameter remains $O(r)$ because the old box already spans the
 system, implying $L=O(r)$.}
 \label{sm:fig:boundary-geometric-cases}
\end{figure}

To replace the spin sum by a flux sum, we need both existence and
multiplicity of the spin realizations. Multiplying the sign patterns
from Lemma~\ref{sm:lem:open-fillings}(ii) component by component realizes
every $E\in\mathfrak P_L^{\mathrm o}$. Every flux has the same number
of realizations: multiplication by one realization gives a bijection
from the flat configurations ($B_p=1$ for all $p$) to the configurations
with that flux. This set of configurations is the \emph{flux fiber}.

It therefore suffices to count flat configurations. Write each sign
as $(-1)^{x_\ell}$ with binary $x_\ell$.
More generally, a binary assignment on $j$-cells is flat if its sum around
the boundary of every retained $(j+1)$-cell is zero modulo two. For $j\ge1$,
the gauge map sends an assignment on $(j-1)$-cells to its boundary sums
on $j$-cells; Eq.~\eqref{sm:cl:gauge-action} is its sign form at $j=m$.
Every flat assignment of positive degree is in this map's image. To see
this, assign to each $(j-1)$-cell the sum of the given values over the
$j$-cells swept out in its full coordinate contraction. The sweep boundary
identity reproduces the original values on every $j$-cell: flatness
cancels the contributions from boundaries of $(j+1)$-cells, and the final
projection has no positive-dimensional cells.

If $z_j$ counts the independent binary choices specifying a flat
assignment on $j$-cells, then $z_0=1$
(the constant assignments) and
$z_j=N_{j-1}^{\mathrm o}-z_{j-1}$ for $j\ge1$: the kernel of the gauge
map on $(j-1)$-cell assignments is exactly their flat subset. Each image has
$2^{z_{j-1}}$ preimages among the $2^{N_{j-1}^{\mathrm o}}$ assignments,
so their ratio gives $2^{z_j}$. Therefore
\begin{equation}
 z_m\coloneqq\sum_{j=0}^{m-1}(-1)^{m-1-j}N_j^{\mathrm o}+(-1)^m,
 \qquad |\text{flux fiber}|=2^{z_m},\qquad
 \frac{z_m}{N}\longrightarrow\binom{d-1}{m-1}.
 \label{sm:cl:fiber}
\end{equation}
The last identity follows from the alternating binomial sum.
In particular $2^{N-1}$ is the fiber size only for $m=1$.

\subsection{Exact gauge-average representation}
\label{sm:cl:gauge-average-section}

Gauge averaging of symmetry-breaking weights has precedents in lattice
gauge theory~\cite{SM-FoersterNielsenNinomiya1980,SM-Testa1997}. Its application
to weak classical perturbations of self-correcting memories is developed
in Ref.~\cite{SM-Thorngren2026}.
The companion analysis~\cite{SM-WatanabeFullPaper2026} uses the same
Boltzmann-weight averaging and charge-cancellation construction.
For thermodynamic analyticity we track the complete partition-function
normalization below; the companion additionally constructs a conditional
inverse sampler for the averaged state.

With the diagonal local representative fixed above, each retained
$v_Y$ is a function of the physical cell signs on $Y$. A function of the signs on $Y$ has the exact binary Fourier expansion
\begin{equation}
 v_Y(\sigma)=\sum_{X\subseteq Y}c_{Y,X}\sigma_X,
 \qquad
 c_{Y,X}\coloneqq2^{-|Y|}\sum_{\sigma|_Y}v_Y(\sigma)\sigma_X.
 \label{sm:cl:local-fourier}
\end{equation}
The sum over $\sigma|_Y$ runs over all sign choices on $Y$, and
$\sigma_\varnothing\coloneqq1$. Orthogonality of these sign products
proves the expansion: the uniform average of their product is one when
the two supports coincide and zero otherwise. The coefficients are real and obey $|c_{Y,X}|\le1$.

Combine equal nonempty supports and keep the constant part explicitly:
\begin{equation}
 V^{\mathrm o}(\sigma)=V_{{\rm c},L}^{\mathrm o}
              +\sum_{X\in\mathcal I_L}c_X\sigma_X.
 \label{sm:cl:fourier-total}
\end{equation}
Here $\mathcal I_L$ contains the distinct nonempty supports with nonzero
combined coefficient; each is one label. The scalar
$V_{{\rm c},L}^{\mathrm o}$ is the sum of the constant Fourier coefficients.
For any $\ell\in X$, at most $n_{\mathrm{inc}}$ original terms can contribute to $c_X$.
Consequently $|c_X|\le n_{\mathrm{inc}}$, $|X|\le b$, diameter at most $R$, and
at most $n_{\mathrm{inc}} 2^{b-1}$ distinct supports contain a specified physical cell.
These bounds also control labels incident on a specified $(m-1)$-cell.
The fixed spatial period gives
$v_{\rm c}\coloneqq\lim_{L\to\infty}N^{-1}V_{{\rm c},L}^{\mathrm o}$.
Restoring this scalar shifts the limiting free energy by $g v_{\rm c}$;
we retain it throughout this proof.

Allow complex $\beta,g$ and define
\begin{equation}
 k_X\coloneqq-\beta gc_X,\qquad
 t_X\coloneqq\tanh k_X,\qquad
 u_{\rm f}\coloneqq e^{-2J\operatorname{Re}\beta},\qquad
 u_{\rm e}\coloneqq\max_X|t_X|.
 \label{sm:an:eq:parameters}
\end{equation}
The coefficients $k_X$ are dimensionless local couplings; $k$ without
a subscript remains the fixed defect dimension. The quantities
$u_{\rm f},u_{\rm e}$ are nonnegative bounds used in estimates, not holomorphic
weights. Set $u_{\rm e}\coloneqq0$ if $\mathcal I_L$ is empty.
Take each $k_X$ in a small disk where $\cosh k_X\ne0$, with
$\log\cosh k_X$ chosen to vanish at zero.

For any fixed $\eta$, permuting $\sigma$ by Eq.~\eqref{sm:cl:gauge-action}
and using invariance of $H_0^{\mathrm o}$ gives the same partition function.
Average these $2^{N_{m-1}^{\mathrm o}}$ equal expressions:
\begin{align}
 w_g(\sigma)&\coloneqq2^{-N_{m-1}^{\mathrm o}}\sum_\eta
                           e^{-\beta gV^{\mathrm o}(\eta\sigma)},
 \label{sm:cl:averaged-weight}\\
 Z_L^{\mathrm o}&=\sum_\sigma e^{\beta J\sum_pB_p(\sigma)}w_g(\sigma).
 \label{sm:cl:averaged-partition}
\end{align}
This averages the full Boltzmann weight, not the perturbation itself.
It is an exact change-of-variables identity for real or complex parameters;
no gauge invariance of $V^{\mathrm o}$ is used.

Use $e^{k_X(\eta\sigma)_X}=\cosh k_X[1+t_X(\eta\sigma)_X]$.
In this finite product a selection $\xi\subseteq\mathcal I_L$ chooses
the second term on precisely its supports. For each label $X$, define its charge set $A_X\coloneqq\partial X$: the $(m-1)$-cells occurring an odd number of times in the boundaries of the cells of $X$. The gauge action gives
$(\eta\sigma)_X=\sigma_X\prod_{u\in A_X}\eta_u$.
Each gauge sign is averaged independently. A sign occurring an odd number
of times makes the average vanish; an even number contributes one.
It follows that
\begin{equation}
 2^{-N_{m-1}^{\mathrm o}}\sum_\eta\prod_{X\in\xi}(\eta\sigma)_X
 =\begin{cases}
 \displaystyle\prod_{X\in\xi}\sigma_X,
    &\text{every gauge cell occurs in an even number of }A_X,\\
 0,&\text{otherwise}.
 \end{cases}
 \label{sm:cl:neutrality-rule}
\end{equation}
Call the retained selections neutral, including the empty one. A label
with $A_X=\varnothing$ is neutral by itself. Define
$\omega_{\rm e}(\xi)\coloneqq\prod_{X\in\xi}t_X$ and
$\sigma_\xi\coloneqq\prod_{X\in\xi}\sigma_X$. Empty products are one.
The exact averaged weight becomes
\begin{align}
 G_L(\sigma)&\coloneqq\sum_{\xi\,\mathrm{neutral}}
                               \omega_{\rm e}(\xi)\sigma_\xi,
 \label{sm:cl:selection-sum}\\
 w_g(\sigma)&=e^{-\beta gV_{{\rm c},L}^{\mathrm o}}
                         \prod_X\cosh k_X\;G_L(\sigma).
 \label{sm:cl:weight-prefactor}
\end{align}
There is no truncation in $g$. In $\sigma_\xi$, repeated physical cells
cancel in pairs, leaving the chain $C_\xi$. Neutrality,
Eq.~\eqref{sm:cl:neutrality-rule}, gives
$\partial C_\xi=\bigtriangleup_{X\in\xi}A_X=\varnothing$.
It may be empty or disconnected.

Fill $C_\xi$ by an interaction cell set $S_{C_\xi}$ using
Lemma~\ref{sm:lem:open-fillings}. Paired boundary physical cells cancel, giving
\begin{equation}
 \sigma_\xi=\prod_{p\in S_{C_\xi}}B_p(\sigma)
           =\chi(\xi,E(\sigma)),
 \label{sm:cl:whole-selection-sign}
\end{equation}
where $\chi(\xi,E)$ is $-1$ if $E$ crosses an odd number of
these interaction cells and $+1$ otherwise. This includes boundary interaction cells.
It is independent of the filling: two fillings differ by a closed
$(m+1)$-chain, which the coordinate sweep fills with $(m+2)$-cells.
The product of its interaction signs is one by
Eq.~\eqref{sm:cl:cube-rule}. This includes boundary-touching components.
No sum over fillings is taken.

The unperturbed factor is also exact. Each negative interaction cell costs $2J$, so
\begin{equation}
 e^{\beta J\sum_pB_p(\sigma)}=e^{\beta JN_p}\omega_{\rm f}(E(\sigma)),
 \qquad\omega_{\rm f}(E)\coloneqq e^{-2\beta J|E|}.
 \label{sm:cl:whole-flux-weight}
\end{equation}
Here $|E|$ counts all occupied dual $k$-cells, including truncated boundary cells.
Insert Eqs.~\eqref{sm:cl:weight-prefactor},
\eqref{sm:cl:whole-selection-sign}, and~\eqref{sm:cl:whole-flux-weight}
into Eq.~\eqref{sm:cl:averaged-partition}. Every allowed flux has exactly
$2^{z_m}$ configurations, as proved above, and their signs in
Eq.~\eqref{sm:cl:whole-selection-sign} agree. Grouping them gives
\begin{equation}
 \boxed{
 Z_L^{\mathrm o}=e^{-\beta gV_{{\rm c},L}^{\mathrm o}}
    2^{z_m}e^{\beta JN_p}\prod_X\cosh k_X
    \sum_{\xi,E}\omega_{\rm e}(\xi)
    \omega_{\rm f}(E)\chi(\xi,E).
 }
 \label{sm:cl:whole-partition}
\end{equation}
The sum is over all neutral selections and all
$E\in\mathfrak P_L^{\rm o}$, each including the empty set.
Every normalization and scalar contribution is retained.

\subsection{Component decomposition of the partition sum}

The whole selection and flux set in Eq.~\eqref{sm:cl:whole-partition}
need not be connected. We denote their nonempty connected perturbation
components by $P_{\rm e}$ and flux components by $P_{\rm f}$:
$\xi=\bigcup_iP_{\rm e,i}$ and $E=\bigcup_jP_{\rm f,j}$.
Each $P_{\rm e}$ or $P_{\rm f}$ denotes one component, not the whole graph.
The subscripts $\mathrm e,\mathrm f$ label the perturbation and flux
species; the label $\mathrm e$ does not introduce an electric term in the
Hamiltonian.

For the convergence estimates, join selected labels whose charge sets
intersect. The resulting graph has bounded degree. A neutral selection decomposes uniquely into connected
neutral components; charges cannot cancel between distinct components.
Charge-empty labels are isolated and give singleton components.
This decomposition is needed to count shapes independently of volume.
For example, two individually neutral interaction cell factors may be arbitrarily
far apart. Fixing the position of one leaves $O(N)$ positions for the other,
although the selection contains only two labels. Each connected component,
in contrast, has spatial diameter controlled by its own cost. Disconnected
flux sets are decomposed for the same reason.

A \emph{perturbation component} $P_{\rm e}$ is a nonempty connected neutral label set.
Its cost is an integer counting microscopic factors, rather than a physical
energy. Define its weight and closed cell chain by
\begin{equation}
 c(P_{\rm e})\coloneqq|P_{\rm e}|,\qquad
 \omega_{\rm e}(P_{\rm e})\coloneqq\prod_{X\in P_{\rm e}}t_X,\qquad
 C_{P_{\rm e}}\coloneqq\mathop{\bigtriangleup}_{X\in P_{\rm e}}X.
 \label{sm:an:eq:electric-data}
\end{equation}
Here $|P_{\rm e}|$ counts selected labels, and the symmetric difference
$\bigtriangleup$ retains physical cells present in an odd number of supports; repeated
physical cells cancel in pairs. By neutrality, $C_{P_{\rm e}}$ is a closed $m$-chain, possibly disconnected
or empty, with at most $b c(P_{\rm e})$ physical cells.
A geometric carrier is a connected region made of lattice cells containing the indicated support.
The microscopic supports have such a carrier of diameter
at most $c_0 c(P_{\rm e})$, for a fixed locality constant $c_0$.

Two perturbation components are incompatible
if they reuse a label or some constituent labels have intersecting charge
sets $A_X$. Mere overlap of their
physical cell supports does not impose incompatibility. Exactly,
\begin{equation}
 G_L(\sigma)=
 \sum_{\{P_{\rm e,1},\ldots,P_{\rm e,r}\}\,\mathrm{compatible}}
       \prod_i \omega_{\rm e}(P_{\rm e,i})\sigma_{C_{P_{\rm e,i}}}.
 \label{sm:an:eq:electric-gas}
\end{equation}
The sum includes the empty collection. Closed singleton labels and
components with empty cell chains are included. Their weights and exclusions contribute to the free-energy density.

A \emph{flux component} $P_{\rm f}$ is a nonempty connected dual-$k$-cell set
satisfying the parity rule of \eqref{sm:cl:cube-rule}. Its cost and weight are
\begin{equation}
 c(P_{\rm f})\coloneqq|P_{\rm f}|,\qquad
 \omega_{\rm f}(P_{\rm f})\coloneqq e^{-2\beta J|P_{\rm f}|}.
 \label{sm:cl:magnetic-data}
\end{equation}
Here $|P_{\rm f}|$ counts dual $k$-cells. Flux components are incompatible if
they overlap or share a non-boundary dual $(k-1)$-face. Every compatible collection of these polymers is a
physical free-box flux, by Lemma~\ref{sm:lem:open-fillings}.

Fix one local physical-interaction cell filling $S_{P_{\rm e}}$ of $C_{P_{\rm e}}$ supplied by
that lemma. Define the mixed sign by the number of occupied interaction cells in
this filling:
\begin{equation}
 \chi(P_{\rm e},P_{\rm f})\coloneqq
 (-1)^{\#\{p\in S_{P_{\rm e}}:\ p\hbox{ is crossed by }P_{\rm f}\}}.
 \label{sm:cl:mixed-sign}
\end{equation}
Boundary interaction cells are included, and filling independence follows by
applying Eq.~\eqref{sm:cl:whole-selection-sign} to this one perturbation
component and the physical flux $P_{\rm f}$.
The magnetic realization $F_{P_{\rm f}}$ gives the reciprocal description
\begin{align}
 \chi(P_{\rm e},P_{\rm f})
 &=\prod_{p\in S_{P_{\rm e}}}B_p(F_{P_{\rm f}})
   =\prod_{\ell\in C_{P_{\rm e}}}(F_{P_{\rm f}})_\ell,\label{sm:an:eq:linking}\\
 \sigma_{C_{P_{\rm e}}}
 &=\prod_{P_{\rm f}\,\mathrm{component\ of}\,E(\sigma)}
                       \chi(P_{\rm e},P_{\rm f}).\label{sm:an:eq:component-product}
\end{align}
These are discrete Stokes identities, valid for every component size.
A negative mixed sign requires intersection with either object's local
filling. No sum over filling choices is taken.
For the component decompositions,
$\chi(\xi,E)=\prod_i\prod_j\chi(P_{\rm e,i},P_{\rm f,j})$.
Both activities also factor over their components. These pair signs and
the same-species exclusions therefore reproduce every summand of
Eq.~\eqref{sm:cl:whole-partition}.

For the common gas notation, either a perturbation component or a flux
component is called a polymer. Let $\mathcal P_L$ contain both species,
with all possible costs. For
same-species pairs put $\chi(P,Q)=0$ or $1$ according to incompatibility or
compatibility. For mixed pairs use the sign in Eq.~\eqref{sm:cl:mixed-sign},
and set $\chi(P,P)=0$. Define $\omega(P)$ by the appropriate species weight above. The symmetric
Mayer factor $h(P,Q)\coloneqq\chi(P,Q)-1$ lies in $\{0,-1,-2\}$
and obeys $|1+h(P,Q)|\le1$.
The full gas is
\begin{equation}
 \Xi_L^{\rm o}\coloneqq\sum_{r\ge0}\frac1{r!}
     \sum_{P_1,\ldots,P_r\in\mathcal P_L}
       \prod_i \omega(P_i)\prod_{i<j}[1+h(P_i,P_j)].
 \label{sm:an:eq:gas}
\end{equation}
Here $r$ counts component occurrences, whereas $c(P_i)$ counts the cells
or selected supports in the $i$th component.
Repeated objects vanish by self-incompatibility. Decomposing the
selections and flux sets in Eq.~\eqref{sm:cl:whole-partition} into their components gives
\begin{equation}
 Z_L^{\rm o}=
 e^{-\beta gV_{{\rm c},L}^{\rm o}}2^{z_m}e^{\beta JN_p}
       \left(\prod_X\cosh k_X\right)\Xi_L^{\rm o}.
 \label{sm:an:eq:prefactor}
\end{equation}
This is the physical free-box partition function exactly, with every
component, scalar term, and normalization factor retained.

\subsection{Convergence and analytic bulk limit}

An exponential cost reserve means that the expansion still converges after
multiplying every cluster weight by an exponential of its total cost.
This reserve pays for later polynomial counting factors and bounds the
large-cost tails. A logarithmic cluster is a tuple connected through
nonzero Mayer factors; it can contain repeated polymer occurrences, whose
costs are counted separately.

\subsubsection{Tree estimate for the linking signs}

A root is an element chosen as a fixed reference for a sum. In the tree
estimate it is one distinguished polymer occurrence. For spatial shape
counts below, it is one support label or dual $k$-cell; the weighted estimates
can instead fix a physical cell.

\begin{lemma}[Cluster convergence with an exponential cost reserve]
\label{sm:lem:tree}
For a self-incompatible polymer gas of the form
\eqref{sm:an:eq:gas}, suppose $h$ is symmetric and parameter independent, $|1+h|\le1$, and each
object has cost $c(P)\ge1$. If some $s>0$ satisfies
\begin{equation}
 \sum_Q |h(P,Q)|\,|\omega(Q)|e^{(1+s)c(Q)}\le c(P)
       \qquad\hbox{for every }P ,
 \label{sm:eq:tree-criterion}
\end{equation}
then the connected expansion of $\log\Xi$ converges absolutely with
reserve $e^{s\sum_i c(P_i)}$. Holomorphic activities with locally uniform
bounds give a holomorphic logarithm and a zero-free partition function.
For infinitely many labels also require
$\sum_P|\omega(P)|e^{(1+s)c(P)}<\infty$ in each finite volume,
locally uniformly in the parameters. Sums can be replaced by integrals against a
positive activity majorant, with the same conclusion.
\end{lemma}

\begin{proof}
\textup{(i) Connected graphs and the tree bound.}
For a tuple of $r$ objects define its connected coefficient
$\mathcal U$ by summing products of $h$ over connected graphs on the
tuple positions; $\mathcal U(P)\coloneqq1$. Grouping graphs into connected
components gives the formal connected-series identity
\begin{equation}
 \log\Xi=
 \sum_{r\ge1}\frac1{r!}\sum_{P_1,\ldots,P_r}
       \mathcal U(P_1,\ldots,P_r)\prod_i \omega(P_i).
 \label{sm:an:eq:cluster-series}
\end{equation}
At this stage the identity is an algebraic statement about coefficients;
neither a branch of the logarithm nor convergence has been assumed. We
will show below that the series is the exact logarithm of the gas.

For a connected graph $G$ on the tuple positions $1,\ldots,r$, choose
position~1 as root. Place each vertex in the layer numbered by its graph
distance from the root. For each nonroot vertex, choose its
smallest-numbered neighbor in the previous layer as parent. The parent edges form a spanning tree $T$. Define $R(T)$ by keeping
every tree edge and every edge within a distance layer. Between consecutive
layers, keep a child-to-previous-layer edge exactly when its endpoint's
number is at least the parent's number; keep no edges across more than
one layer. A graph producing $T$ therefore satisfies $T\subseteq G\subseteq R(T)$ [Fig.~\ref{sm:fig:partition-example}].
Conversely, these inclusions preserve the distance layers, since no added
edge skips a layer, and preserve each smallest-numbered parent by the
endpoint rule. Thus the intervals partition all connected graphs.
Summing over each interval gives the Penrose identity
\cite{SM-JacksonProcacciSokal2014},
\begin{equation}
 \mathcal U(P_1,\ldots,P_r)=
 \sum_T\prod_{\{i,j\}\in \operatorname{Edges}(T)}h(P_i,P_j)
        \prod_{\{i,j\}\in \operatorname{Edges}(R(T))\setminus \operatorname{Edges}(T)}
                      [1+h(P_i,P_j)] .
 \label{sm:an:eq:penrose}
\end{equation}
Here $\operatorname{Edges}(T)$ denotes the tree-edge set. Thus $|\mathcal U|$ is bounded
by the sum of tree products of $|h|$, including when a linking pair has
$h=-2$. Taking absolute values before this connected-graph cancellation
would not give the required estimate.

\begin{figure}[htbp]
 \centering
 \includegraphics{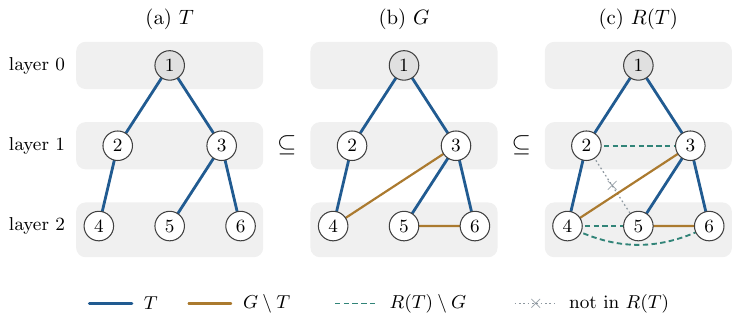}
 \caption{\textbf{Connected-graph partition for $r=6$.}
 Vertices label the six polymer occurrences, and gray bands indicate graph
 distance from root~1; inclusions refer to edge sets.
 (b)~In $G$, vertex~4 has neighbors~2 and~3 in the previous layer and
 chooses~2 as parent; vertices~5 and~6 have only~3.
 (a)~These parent edges form $T$ (blue).
 (c)~$R(T)$ adds all edges within a layer and the edge $\{3,4\}$, since
 $3$ is larger than vertex~4's parent~2.
 Orange edges belong to $G\setminus T$, and dashed teal edges to $R(T)\setminus G$.
 The crossed gray edge $\{2,5\}$ is excluded: adding it would change
 vertex~5's parent from~3 to~2. Edges from root~1 to layer~2 are also
 excluded because they skip a layer.}
 \label{sm:fig:partition-example}
\end{figure}

\textup{(ii) Sum the rooted trees.}
Rooted-tree sums of height at most $d_{\rm t}$, where
$d_{\rm t}=0,1,\ldots$, are majorized by
\begin{equation}
 T_0(P)\coloneqq1,\qquad
 T_{d_{\rm t}+1}(P)\coloneqq\exp\!\left[\sum_Q|h(P,Q)||\omega(Q)|e^{s c(Q)}T_{d_{\rm t}}(Q)\right]
                    \le e^{c(P)} .
 \label{sm:an:eq:tree-recursion}
\end{equation}
The last inequality follows inductively from
\eqref{sm:eq:tree-criterion}; the exponential counts any number of
children with their factorials. Consequently,
\begin{equation}
 \sum_{r\ge1}\frac1{(r-1)!}\sum_{P_2,\ldots,P_r}
  |\mathcal U(P,P_2,\ldots,P_r)|
       \prod_{i=2}^r |\omega(P_i)|e^{s c(P_i)}
       \le e^{c(P)} .
 \label{sm:an:eq:pinned}
\end{equation}
Multiply by $|\omega(P)|e^{s c(P)}$ and sum roots to obtain absolute
convergence. Locally uniform majorants give holomorphy in complex
parameters when all activities are holomorphic.

\textup{(iii) Identify the exact logarithm.}
To identify the sum with the original gas, introduce an auxiliary complex
number $z$ and replace every activity by $z\omega(P)$. Let
$\Xi(z)$ denote the resulting gas and $\mathcal L(z)$ the
connected series on the right-hand side of
\eqref{sm:an:eq:cluster-series} with this replacement.
The gas itself is absolutely convergent: since each pair factor has
modulus at most one, its absolute sum is bounded by
$\exp(|z|\sum_P|\omega(P)|)$. The finite-volume summability assumption
makes this bound finite, also for a countable polymer set.
The same bounds give
absolute convergence for $|z|\le1$ and holomorphy for $|z|<1$.

Near $z=0$, graph decomposition gives
$\exp\mathcal L(z)=\Xi(z)$ as a convergent identity.
The identity theorem extends it throughout $|z|<1$, and absolute
convergence permits $z\to1$. Thus
$\Xi=\exp\mathcal L(1)$ exactly. In particular $\Xi\ne0$, and
$\mathcal L(1)$ supplies a holomorphic logarithm wherever the parameter
bounds hold locally uniformly. This step connects the convergent series
to the unexpanded partition function. The argument is a cost-weighted form of
the familiar polymer criterion \cite{SM-KoteckyPreiss1986}.
Its extension to complex pair interactions satisfying $|1+h|\le1$
was established in Ref.~\cite{SM-Ueltschi2004}; the proof above records
the exponential cost reserve needed here.
For purely hard-core gases, Fern\'andez and Procacci give a sharper
criterion~\cite{SM-FernandezProcacci2007}; the tree estimate above directly
accommodates our mixed crossing signs.
\end{proof}

\subsubsection{Uniform geometric bounds for all component sizes}

There are constants $A_{\rm f},A_{\rm e}$ such that dual flux sets of size $r$ through a
specified dual $k$-cell number at most $A_{\rm f}^r$, and perturbation components of cost
$j$ through a specified label number at most $A_{\rm e}^j$. Encode a connected
cell or label set by choosing a spanning tree deterministically and
traversing each tree edge in both directions from the root. The visited
vertices recover the set, and the traversal length is at most twice its
size. Since each step has a bounded number of choices, the number of
sets grows at most exponentially. This applies to both adjacency graphs,
including at the physical boundary; imposing neutrality or the parity
rule only reduces their counts.
The activity bounds are
\begin{equation}
 |\omega_{\rm f}(P_{\rm f})|\le u_{\rm f}^{c(P_{\rm f})},\qquad
 |\omega_{\rm e}(P_{\rm e})|\le u_{\rm e}^{c(P_{\rm e})}.
 \label{sm:cl:activity-bounds}
\end{equation}
Same-species exclusions give root counts linear in the fixed object's
cost: sum over its dual faces or its labels and charged $(m-1)$-cells.

The mixed counts require the reciprocal identities in
\eqref{sm:an:eq:linking}. Fix a perturbation component of cost $j$ and vary a
dual flux set of size $r$. The cell product in Eq.~\eqref{sm:an:eq:linking}
shows that $\chi=-1$ requires $(F_{P_{\rm f}})_\ell=-1$ for at least one
physical cell $\ell\in C_{P_{\rm e}}$. By the local magnetic filling, a deterministic root cell
of $P_{\rm f}$ is therefore within distance $O(r)$ of one of at most $bj$
physical cells of the closed chain. Conversely, fixing a dual flux set of size $r$, a negative
sign with a perturbation component of cost $j$ implies that an interaction cell of
$S_{P_{\rm e}}$ is crossed by one of its $r$ cells. A perturbation root label must
then be within distance $O(j)$ of that cell. Thus the possible root
positions satisfy
\begin{align}
 \#\{\hbox{flux roots of size }r\hbox{ with }\chi(P_{\rm e},P_{\rm f})=-1\}
       &\le Cj(r+1)^d,\label{sm:cl:flux-root-count}\\
 \#\{\hbox{perturbation roots of cost }j\hbox{ with }\chi(P_{\rm e},P_{\rm f})=-1\}
       &\le Cr(j+1)^d.\label{sm:cl:mixed-root-counts}
\end{align}
These estimates are uniform for boundary components and components as
large as the box. Counting possible root positions using the varying
object's local filling keeps the bound linear in the fixed object's cost.

For a fixed reserve $s>0$, put $r_{\rm f}\coloneqq A_{\rm f}u_{\rm f}e^{1+s}$ and
$r_{\rm e}\coloneqq A_{\rm e}u_{\rm e}e^{1+s}$. Combining exclusions and mixed signs, for which
$|h|=2$, gives
\begin{align}
 \frac{\sum_P|h(P_{\rm f},P)||\omega(P)|e^{(1+s)c(P)}}{c(P_{\rm f})}
 &\le C_{\rm ff}\sum_{r\ge1}r_{\rm f}^r+
        C_{\rm fe}\sum_{j\ge1}j^dr_{\rm e}^j,\label{sm:an:eq:flux-row-bound}\\
 \frac{\sum_P|h(P_{\rm e},P)||\omega(P)|e^{(1+s)c(P)}}{c(P_{\rm e})}
 &\le C_{\rm ee}\sum_{j\ge1}r_{\rm e}^j+
        C_{\rm ef}\sum_{r\ge1}r^dr_{\rm f}^r.\label{sm:an:eq:row-bounds}
\end{align}
In particular, take $s=1$ and choose $D$ no smaller than the four
constants in Eqs.~\eqref{sm:an:eq:flux-row-bound} and~\eqref{sm:an:eq:row-bounds}. Since $r^d\ge1$,
a convenient sufficient condition for both rows is
\begin{equation}
 D\sum_{r\ge1}r^d\left[(A_{\rm e}e^2u_{\rm e})^r
                      +(A_{\rm f}e^2u_{\rm f})^r\right]<1.
 \label{sm:cl:simple-convergence}
\end{equation}
This is an estimate on the absolute weights, not a redefinition of the
holomorphic component weights. The two geometric sums become small
when $u_{\rm e}$ is small and $J\operatorname{Re}\beta$ is large.

Taking $J\operatorname{Re}\beta$ sufficiently large and $u_{\rm e}$
sufficiently small makes both rows strictly less than one, with thresholds
depending only on the interaction bounds and the chosen reserve.
Lemma~\ref{sm:lem:tree} therefore gives a holomorphic logarithm of the exact
gas with locally uniform rooted bounds independent of $L$.
Increasing $J\operatorname{Re}\beta$ and decreasing $u_{\rm e}$ allows any
fixed larger reserve.

More explicitly, fix any allowed real $(\beta_0,g_0)$. Since the row bounds
are strict, continuity gives an open complex neighborhood of this point
on which they still hold with a margin. This neighborhood can be chosen
independently of $L$: only the fixed locality bounds enter the estimates.
The activities, $\cosh$ factors, and chosen $\log\cosh$ factors are
holomorphic there. All compact subsets therefore have common rooted
majorants for every finite volume.

\subsubsection{The thermodynamic limit and its common cutoff argument}
\label{sm:bulk-limit}

Let $\mathcal C$ denote a cluster tuple in
\eqref{sm:an:eq:cluster-series}, with
\begin{equation}
 \operatorname{cost}(\mathcal C)\coloneqq\sum_{P\in\mathcal C}c(P),\qquad
 \phi(\mathcal C)\coloneqq
 \frac{\mathcal U(\mathcal C)}{|\mathcal C|!}
                         \prod_{P\in\mathcal C}\omega(P).
 \label{sm:cl:cluster-cost}
\end{equation}
Multiply Eq.~\eqref{sm:an:eq:pinned} by
$|\omega(P)|e^{s c(P)}$ and sum the root. An $r$-occurrence cluster
is then counted $r$ times, so this bounds the unrooted sum:
\begin{equation}
 \sum_{\mathcal C}|\phi(\mathcal C)|e^{s\operatorname{cost}(\mathcal C)}
 \le\sum_P|\omega(P)|e^{(1+s)c(P)}\le CN.
 \label{sm:cl:reserved-total}
\end{equation}
The last inequality follows by summing the rooted shape bounds over the
$O(N)$ possible spatial roots. On total cost at least $M$, the extra
weight is at least $e^{sM}$. Hence, on compact subsets of the strict
convergence domain,
\begin{align}
 \frac1N\sum_{\mathcal C:\operatorname{cost}(\mathcal C)\ge M}
                  |\phi(\mathcal C)|&\le Ce^{-sM},\label{sm:an:eq:bulk}\\
 \psi_{\rm pol}(\beta J,\{k_X\})&\coloneqq
             \lim_{L\to\infty}\frac1N\log\Xi_L^{\rm o}.\label{sm:an:eq:bulk-definition}
\end{align}
Here $s>0$ is the reserve already chosen above, and $|\mathcal C|$
counts occurrences, including repetitions.

The following cutoff argument
also applies to the later weighted and high-temperature expansions.
Suppose the cluster coefficients are holomorphic on a common complex
neighborhood and have a summable bound per spatial root, uniformly in
volume on compact subsets. At each fixed total cost, suppose the bulk
cluster families and coefficients agree with a fixed periodic infinite-volume family, apart
from a boundary layer of fixed width. Suppose also that the normalized
cluster sums have a volume-uniform exponential cost tail on every compact
subset of a common complex neighborhood. Then the fixed-cost sums per
cell converge to the average of their bulk translations over one spatial
period: the boundary layer contributes only $O(L^{d-1})/N$. For periodic
boxes the same conclusion holds once fixed-cost clusters cannot wrap.
At fixed cost there are finitely many geometric shapes relative to a
root; if a shape carries countably many labels, a common absolutely
summable rooted majorant suffices to sum those labels first.

First choose $M$ to make the exponential tail uniformly small on a fixed
complex compact set. Then take $L\to\infty$ in the costs below $M$,
and finally let $M\to\infty$. The convergence is uniform on that compact
set, so the limit of the holomorphic normalized sums is holomorphic.
This argument controls the full series without changing any finite-volume
partition function.

For the present classical gas, successive objects of a connected cluster
either exclude each other or have a negative mixed sign. Their carriers
and local fillings bound its spatial extent by a constant times total
cost. Thus fixed-cost clusters meeting the boundary have roots in a layer
of fixed width; all remaining shapes and coefficients are those of the
fixed periodic interaction. Equation~\eqref{sm:an:eq:bulk} supplies the
required uniform tail, proving the stated limit for $\psi_{\rm pol}$.
Real-axis convergence alone would not suffice.\cite{SM-Note1} The common complex neighborhood and locally uniform
convergence there are what preserve analyticity.

The full limiting logarithm per cell is $\psi^{\rm o}$,
equal to $-\beta$ times the free-energy density. Since $z_m/N\to\binom{d-1}{m-1}$ and
$N_p/N\to\binom d{m+1}$, the exact normalization gives
\begin{equation}
 \begin{aligned}
 \psi^{\rm o}(\beta,g)\coloneqq\lim_{L\to\infty}\frac1N\log Z_L^{\rm o}
 ={}&-\beta g v_{\rm c}+\binom{d-1}{m-1}\log2+\binom d{m+1}\beta J\\
 &+\overline{\sum_X\log\cosh(-\beta gc_X)}
   +\psi_{\rm pol}(\beta J,\{-\beta gc_X\}).
 \end{aligned}
 \label{sm:an:eq:pressure-formula}
\end{equation}
The overline denotes a sum with one representative of each support orbit
under translations by the fixed spatial period, divided by the number
of unit cells in that period. This convention counts supports crossing a
period boundary once as well. Scalars and $\cosh$ factors remain explicitly in this
formula. Because $|c_X|\le n_{\mathrm{inc}}$, a sufficiently small bound on
$\beta|g|$ implies the required bound on $u_{\rm e}$. Each allowed real point
therefore has a complex neighborhood where the right-hand side is
holomorphic.

\subsection{Boundary independence and the periodic theorem}
\label{sm:cl:boundary-comparison}
\label{sm:an:sec:real-bridge}

At real parameters, compare the original periodic Hamiltonian with the
free-box Hamiltonian for the same interaction. The periodic lattice has
$\Delta N_m\coloneqq\binom dmL^d-N_m^{\mathrm o}=O(L^{d-1})$
physical cells absent from the free box. Embed $H_g^{\mathrm o}$ by adding
$\Delta N_m$ spectator qubits, tensoring with their identity $I$. They contribute exactly
$2^{\Delta N_m}$ to its partition function.
Only interaction cells and perturbing terms crossing a boundary seam differ.
Finite range and bounded incidence give, with a volume-independent constant $C_{\partial}$,
\begin{equation}
 \left\|H_g-H_g^{\rm o}\otimes I\right\|
                  \le C_{\partial}(J+|g|)L^{d-1}.
 \label{sm:cl:boundary-norm}
\end{equation}
This estimate is valid also for general Hermitian quantum perturbations.

For Hermitian matrices $H,H'$ on the same space, their ordered eigenvalues
differ by at most $\|H-H'\|$. Summing their real Boltzmann weights gives
\begin{equation}
 \left|\log\operatorname{Tr}e^{-\beta H}
       -\log\operatorname{Tr}e^{-\beta H'}\right|
                  \le\beta\|H-H'\|\qquad(\beta>0).
 \label{sm:cl:trace-comparison}
\end{equation}
No commutativity is needed. Applying this to the embedded free Hamiltonian
and dividing by $N=L^d$ yields
\begin{equation}
 \left|\frac1N\log Z_L
              -\frac1N\log Z_L^{\rm o}\right|
       \le\frac{C_{\partial}\beta(J+|g|)}{L}+\frac{\Delta N_m}{N}\log2\ \longrightarrow\ 0.
 \label{sm:an:eq:real-comparison}
\end{equation}
Thus $-\beta$ times the original periodic free-energy density equals
$\psi^{\rm o}(\beta,g)$ at every
allowed real point. The holomorphic free-box limit constructed from the
exact gas is consequently a holomorphic extension of this periodic
dimensionless logarithm per cell. Its restriction to real parameters is jointly
real analytic, and $f=-\psi^{\rm o}/\beta$ proves
the classical statement of Theorem~\ref{sm:thm:low-temperature}.
The boundary comparison is used at real parameters; the holomorphic
extension has already been established by the exact free-box expansion.

\subsection{Weighted interactions for the quantum extension}
\label{sm:cl:weighted}

The quantum reduction generates interactions with decaying spatial tails.
The retention of generating carriers and exponential cost bounds is shared
with the weighted classical expansion in Ref.~\cite{SM-WatanabeFullPaper2026}.
The lemma below includes the exact scalar and $\cosh$ normalization and a
locally uniform complex bulk limit. The companion's weighted transport
statement additionally requires inverse gauge sampling and defect
decoder/encoder estimates.

We record a weighted version of the exact gas argument. Fix a positive
exponential-weight parameter $S$. On the same free box, consider the
dimensionless diagonal exponent $a_{0,L}+\sum_\alpha k_{\alpha,L}\sigma_{Y_\alpha}$,
and define its volume-uniform weighted coefficient bound by
\begin{equation}
 Q_S\coloneqq\sup_L\sup_\ell\sum_{\alpha:\ell\in Y_\alpha}
                    |k_{\alpha,L}|e^{S\nu_\alpha}.
 \label{sm:cl:weighted-norm}
\end{equation}
Every $Y_\alpha$ is nonempty and has integer cost $\nu_\alpha\ge1$,
with $|Y_\alpha|\le b_w\nu_\alpha$ and a connected geometric carrier of
diameter at most $d_w\nu_\alpha$. The constants $b_w,d_w$ are fixed.
Different labels can carry the same cell product; finite or countably
infinite alphabets are allowed. The scalar $a_{0,L}$ is kept separately.
All bounds hold uniformly in $L$ on compact subsets of a common complex
parameter neighborhood.

The finite-range classical model is recovered by
$\alpha=X$, $Y_\alpha=X$, $k_{\alpha,L}=k_X$, $\nu_\alpha=1$, and
$a_{0,L}=-\beta gV_{{\rm c},L}^{\mathrm o}$; for fixed $S$ it has
$Q_S=O(|\beta g|)$. Define
$t_{\alpha,L}\coloneqq\tanh k_{\alpha,L}$ and define the charge set $A_\alpha\coloneqq\partial Y_\alpha$, again with the boundary taken modulo two. Neutral selections and their
components use the same charge-cancellation and overlap rules as above;
their perturbation weights are products of $t_{\alpha,L}$. Their closed cell chain is $C_{P_{\rm e}}\coloneqq\bigtriangleup_{\alpha\in P_{\rm e}}Y_\alpha$; neutrality again gives empty boundary.

\begin{lemma}[Exact weighted free-box gas]
\label{sm:lem:weighted-gas}
Assume the coefficients $k_{\alpha,L}$ are holomorphic in their complex
parameters, with locally uniform bounds.
Fix an exponential reserve $s>0$. For sufficiently large fixed $S$, small
$Q_S$, and large $J\operatorname{Re}\beta$, the resulting two-species expansion
satisfies Lemma~\ref{sm:lem:tree}, with perturbation
cost $c(P_{\rm e})\coloneqq\sum_{\alpha\in P_{\rm e}}\nu_\alpha$. Its partition function is
exactly
\begin{equation}
 \begin{aligned}
 &\sum_\sigma e^{\beta J\sum_pB_p(\sigma)+a_{0,L}
                         +\sum_\alpha k_{\alpha,L}\sigma_{Y_\alpha}}\\
 &\qquad=e^{a_{0,L}}2^{z_m}e^{\beta JN_p}
                   \left(\prod_\alpha\cosh k_{\alpha,L}\right)\Xi_L^{\rm o}.
 \end{aligned}
 \label{sm:cl:weighted-prefactor}
\end{equation}
Its connected gas series has the uniform total-cost tail
\eqref{sm:an:eq:bulk}. If the fixed-cost label families, carriers, and coefficients stabilize
in the bulk with a fixed spatial period, and $N^{-1}a_{0,L}$ has a locally uniformly
holomorphic limit, the normalized logarithm per cell also has such a limit.
\end{lemma}

\begin{proof}
\textup{(i) Weighted counting of perturbation components.}
The gauge cancellation and the physical free-box fillings are unchanged.
We must replace the bounded-degree perturbation counting by a weighted bound.
For sufficiently small $k_{\alpha,L}$, $|t_{\alpha,L}|\le2|k_{\alpha,L}|$,
also for complex coefficients.
A charge set has at most $2m b_w\nu_\alpha$ gauge cells. At a fixed
$(m-1)$-cell, sum labels through one of its at most $2(d-m+1)$ incident
physical $m$-cells. Write $\alpha\sim \alpha'$ when $\alpha=\alpha'$ or $A_\alpha\cap A_{\alpha'}\ne\varnothing$.
With this adjacency, the bound is
\begin{equation}
 \sum_{\alpha'\sim \alpha}|t_{\alpha',L}|e^{S\nu_{\alpha'}}
                  \le C b_w Q_S\nu_\alpha.
 \label{sm:cl:charge-majorant}
\end{equation}
For a charge-empty label only the coincident-label contribution is needed;
it obeys the same bound by choosing any physical cell in its nonempty support.

Choose a prescribed reserve $A>1+s$ and then $S\ge A+1$. Connected label
sets rooted at $\alpha$ are bounded by rooted trees, with repeated
children allowed to increase the majorant. Their sums obey
\begin{equation}
 T_\alpha=|t_{\alpha,L}|e^{A\nu_\alpha}
                \exp\!\left(\sum_{\alpha'\sim \alpha}T_{\alpha'}\right),\qquad
 T_\alpha\le |t_{\alpha,L}|e^{(A+1)\nu_\alpha}
                   \quad\hbox{if }C b_w Q_S\le1.
 \label{sm:cl:tree-majorant}
\end{equation}
Starting from $T_\alpha^{(0)}=0$, monotone iteration of the right-hand side
sums finite rooted trees and preserves the displayed upper bound.
The exponential counts unordered children with their factorials. Restricting
to neutral sets decreases the sum. Rooting a connected selection at a label
whose support contains $\ell$ and summing such labels gives
\begin{equation}
 \sup_\ell\sum_{P_{\rm e}:\ell\in\bigcup_{\alpha\in P_{\rm e}}Y_\alpha}
                   |\omega_{\rm e}(P_{\rm e})|e^{A c(P_{\rm e})}\le C_AQ_S.
 \label{sm:cl:weighted-electric-root}
\end{equation}
A selection containing several such labels is overcounted, which is harmless.
This also bounds selections with a specified deterministic physical cell
root.

\textup{(ii) Control the mixed and exclusion terms.}
Closed perturbation chains have $O(c(P_{\rm e}))$ physical cells and carrier diameter
$O(c(P_{\rm e}))$. Lemma~\ref{sm:lem:open-fillings} and the reciprocal mixed
root counts therefore apply at every cost. In a row with fixed flux
cost $r$, varying perturbation cost $j$ incurs at most $Cr(j+1)^d$ root
positions. At cost $j$, replacing the weight $e^{Aj}$ in
Eq.~\eqref{sm:cl:weighted-electric-root} by $e^{(1+s)j}$ gains
$e^{-(A-1-s)j}$. Thus
\begin{equation}
 \begin{aligned}
 \sum_{P_{\rm e}}|h(P_{\rm f},P_{\rm e})|
       |\omega_{\rm e}(P_{\rm e})|e^{(1+s)c(P_{\rm e})}
 &\le C_A rQ_S\sum_{j\ge1}(j+1)^de^{-(A-1-s)j}\\
 &\le C_{A,s}rQ_S.
 \end{aligned}
 \label{sm:cl:weighted-mixed-row}
\end{equation}
The scalar series is finite because $A>1+s$; $C_{A,s}$ is independent
of volume. Perturbation exclusions are
bounded by the $O(c(P))$ support physical cells and charged $(m-1)$-cells of the fixed
object, using the same rooted bound.

Using the already established bound $A_{\rm f}^r$ for flux shapes of
size $r$ through a fixed cell, the resulting row estimate is
\begin{equation}
 \sum_Q|h(P,Q)||\omega(Q)|e^{(1+s)c(Q)}
 \le C\left[Q_S+
     \sum_{r\ge1}r^d(A_{\rm f}u_{\rm f}e^{1+s})^r\right]c(P).
 \label{sm:cl:weighted-row}
\end{equation}
The bracket is sufficiently small in the stated domain, establishing the
criterion. The total root activity in a finite box is finite by
\eqref{sm:cl:weighted-electric-root} and magnetic counting. Thus
Lemma~\ref{sm:lem:tree} also applies to the countable alphabet.

\textup{(iii) Recover the exact normalization and bulk limit.}
Absolute summability of $\sum_\alpha|k_{\alpha,L}|$ gives the exact gauge
product and the product of $\cosh$ factors. One can first prove the identities
for finite label sets and then pass to the full alphabet using absolute
convergence; this proves \eqref{sm:cl:weighted-prefactor} with every label
retained.

For the bulk limit, stabilization at fixed cost means that the label
families, carriers, and coefficients agree, under a translation-preserving
identification, with a fixed periodic family once their carriers are
farther from the boundary than a constant times their cost. The constant
is uniform over labels and volumes. This condition concerns the whole
fixed-cost family, not only each coefficient separately. The rooted bounds
justify summing any countably many labels of each fixed geometric shape.
Together with the exponential cost tail, these verify the hypotheses of
the common cutoff argument in Sec.~\ref{sm:bulk-limit}.
The $\cosh$ contribution is absolutely summable,
since $|\log\cosh k_{\alpha,L}|\le C|k_{\alpha,L}|^2$ in the chosen disk; its
weighted coefficient bound supplies the same limiting argument. Together
with the retained scalar limit this proves the bulk-limit assertion.
\end{proof}
\section{Quantum proof}
\label{sm:sec:quantum}
\label{sm:qu:section}

We prove the quantum part of Theorem~\ref{sm:thm:low-temperature},
with the model and locality assumptions stated in
Section~\ref{sm:cl:section} and the auxiliary open cube defined in
Section~\ref{sm:cl:open-box}. Thus
$H_g^{\mathrm o}=H_0^{\mathrm o}+g\sum_Xv_X$, where every retained
$v_X$ is Hermitian, has $\norm{v_X}\le1$, and acts on its assigned
nonempty support $X$ of at most $b$ physical $m$-cells and diameter at most $R$.
At most $n_{\mathrm{inc}}$ original supports meet each physical cell.
We keep each original $v_X$ in full, including its identity part, so
its norm bound is preserved. Constant contributions will be included
in the local logarithm below, without a separate scalar subtraction.

Imaginary-time expansions and effective classical interactions are established tools for quantum perturbations of classical systems~\cite{SM-BorgsKoteckyUeltschi1996,SM-KoteckyUeltschi1999,SM-FrohlichReyBelletUeltschi2001}. The construction below supplies the weighted, volume-uniform estimates needed here.

The diagonal matrix elements of the Boltzmann operator give an exact
classical sum with weak interactions of decaying spatial range.
We allow both $\beta$ and $g$ to be complex near each allowed real point.
The final effective interactions are indexed by their support $Y$, as in
the classical representation. Only inside the estimates do we retain an
auxiliary label recording the cluster that generates each interaction.
This representation evaluates the trace; it does not assert
that the quantum Gibbs state is diagonal. No spectral gap is assumed.

The proof has three steps. First, an imaginary-time expansion gives the
exact diagonal ratio and hence the partition function. Grouping the
perturbation labels in each history by overlaps of their enlarged supports
gives a polymer gas. Expanding its logarithm produces weak classical
interactions with decaying spatial tails. Second, we apply
Lemma~\ref{sm:lem:weighted-gas}, retaining all scalar and $\cosh$ factors.
We check that fixed-cost contributions stabilize in the periodic bulk
and that the remaining tails vanish uniformly; these properties give
a locally uniform analytic limit of the free-energy density in a free box.
Third, a boundary comparison of the original Hamiltonians at real
parameters shows that the periodic free energy has the same limit.

The imaginary-time diagonal reduction and control of its generating
supports are shared with Ref.~\cite{SM-WatanabeFullPaper2026}. Here we
additionally prove holomorphic coefficient bounds, keep the scalar
contribution to the trace, and verify the analytic bulk limit. The companion
uses the real diagonal interaction as input to its channels and separately
constructs local coherence recovery.

\subsection{Step 1: exact local logarithmic coefficients}
\label{sm:qu:shadow}

Let $\ket\sigma$ be the tensor-product eigenvector with
$Z_{\ell}\ket\sigma=\sigma_{\ell}\ket\sigma$, $\sigma_{\ell}=\pm1$,
and set $H_0^{\mathrm o}(\sigma)\coloneqq-J\sum_{p\subset\Lambda_L}B_p(\sigma)$.
Define the full diagonal ratio by
\begin{equation}
 \mathcal W_{\beta,g}(\sigma)\coloneqq
 e^{\beta H_0^{\mathrm o}(\sigma)}
 \bra\sigma e^{-\beta H_g^{\mathrm o}}\ket\sigma.
 \label{sm:qu:diagonal-ratio}
\end{equation}
Because $H_0^{\mathrm o}\ket\sigma=H_0^{\mathrm o}(\sigma)\ket\sigma$,
the unperturbed diagonal Boltzmann weight is
$\bra\sigma e^{-\beta H_0^{\mathrm o}}\ket\sigma
=e^{-\beta H_0^{\mathrm o}(\sigma)}$.
Thus $\mathcal W_{\beta,g}(\sigma)$ is the ratio of the full diagonal weight to
this unperturbed weight, and $\mathcal W_{\beta,0}(\sigma)=1$.
Write the trace as
$Z_L^{\mathrm o}=\sum_\sigma\bra\sigma e^{-\beta H_g^{\mathrm o}}\ket\sigma$
and substitute Eq.~\eqref{sm:qu:diagonal-ratio}. This gives
\begin{equation}
 Z_L^{\mathrm o}(\beta,g)=
 \sum_\sigma e^{-\beta H_0^{\mathrm o}(\sigma)}\mathcal W_{\beta,g}(\sigma).
 \label{sm:qu:full-trace}
\end{equation}
Thus the desired classical representation starts from the full quantum
partition function, not a truncated perturbation series. The next task is
to construct a local logarithm of $\mathcal W_{\beta,g}$ and control it uniformly
in volume.

Define the local imaginary-time operators
\begin{equation}
 v_X(u)\coloneqq e^{u\beta H_0^{\mathrm o}}v_X e^{-u\beta H_0^{\mathrm o}},
 \qquad 0\le u\le1.
 \label{sm:qu:complex-shadow}
\end{equation}
Here $u$ is the dimensionless imaginary-time variable, distinct from
the gauge-cell index in $\eta_u$: the physical imaginary time is $\beta u$
when $\beta$ is real.
We now derive the time-ordered formula from the definition of
$\mathcal W_{\beta,g}$, keeping the operator order at every step.
First define
\begin{equation}
 W(u)\coloneqq e^{u\beta H_0^{\mathrm o}}e^{-u\beta H_g^{\mathrm o}}.
 \label{sm:qu:interaction-evolution}
\end{equation}
It obeys $W(0)=I$, where $I$ is the identity operator.
The product rule and $H_g^{\mathrm o}-H_0^{\mathrm o}=g\sum_Xv_X$ give
\begin{equation}
 \begin{aligned}
 \frac{dW(u)}{du}
 &=\beta e^{u\beta H_0^{\mathrm o}}
       (H_0^{\mathrm o}-H_g^{\mathrm o})e^{-u\beta H_g^{\mathrm o}}\\
 &=-\beta g\sum_Xv_X(u)W(u).
 \end{aligned}
 \label{sm:qu:interaction-derivative}
\end{equation}
In the second line we inserted
$e^{-u\beta H_0^{\mathrm o}}e^{u\beta H_0^{\mathrm o}}=I$ between
$v_X$ and the final exponential; no commutativity of the perturbation
is required. Integration gives
\begin{equation}
 W(u)=I-\beta g\sum_X\int_0^u du_1\,v_X(u_1)W(u_1).
 \label{sm:qu:interaction-integral}
\end{equation}

Repeatedly substituting this formula for the remaining $W(u_1)$
produces the ordered times $0<u_n<\cdots<u_1<u$, with the latest time
on the left:
\begin{equation}
 \begin{aligned}
 W(u)={}&I+\sum_{n\ge1}(-\beta g)^n
 \sum_{X_1,\ldots,X_n}
 \int_{0<u_n<\cdots<u_1<u}du_1\cdots du_n\\
 &\quad\times v_{X_1}(u_1)\cdots v_{X_n}(u_n).
 \end{aligned}
 \label{sm:qu:interaction-series}
\end{equation}
Labels may repeat, and every order is retained. The ordered integration
region has volume $u^n/n!$, so for $u\le1$ the sum of the norms is bounded by
\begin{equation}
 \begin{aligned}
 &\sum_{n\ge0}\frac{|\beta g|^n}{n!}
   \left[\sum_X\sup_{0\le s\le1}\norm{v_X(s)}\right]^n\\
 &\qquad=\exp\left[|\beta g|\sum_X\sup_{0\le s\le1}\norm{v_X(s)}\right]<\infty.
 \end{aligned}
 \label{sm:qu:interaction-majorant}
\end{equation}
The remainder after repeated substitution is bounded by the next-order
term of this majorant times $\sup_{0\le s\le1}\norm{W(s)}$, and tends
to zero. Thus the series is exact for every finite volume and finite
complex $\beta,g$. The exponent in this bound contains a sum over all
supports and grows with volume; the local bounds below will instead give
a convergence domain independent of volume.

The series defines the time-ordered exponential, with $\mathcal T$
placing larger times to the left. No extra $1/n!$ multiplies an ordered
integral. Setting $u=1$ gives
$\bra\sigma W(1)\ket\sigma=\mathcal W_{\beta,g}(\sigma)$ directly from
Eq.~\eqref{sm:qu:diagonal-ratio}, and hence the rescaled Duhamel identity
\begin{equation}
 \boxed{
 \mathcal W_{\beta,g}(\sigma)=
 \bra\sigma\mathcal T\exp\left[-\beta g\int_0^1du
                  \sum_X v_X(u)\right]\ket\sigma.
 }
 \label{sm:qu:duhamel}
\end{equation}

Since the unperturbed operators $B_p$ commute, only the at most
$2(d-m)b$ interaction cells touching the original support of $v_X$ enter
its conjugation. All other interaction cell exponentials
cancel. Define the fixed enlarged support $U_X\coloneqq X\cup\bigcup_{p:\partial p\cap X\ne\varnothing}\partial p$, where the union uses only retained interaction cells. It contains the support of $v_X(u)$ for every $u,\beta,g$, and has at most $b_h\coloneqq b[1+2(d-m)(2m+1)]$ physical $m$-cells and bounded diameter. The touching-interaction cell Hamiltonian has norm at most $2(d-m)bJ$.
For Hermitian $A$, $\norm{e^{zA}}\le e^{|\operatorname{Re}z|\norm A}$;
applying this to the two conjugating exponentials and using
$\norm{v_X}\le1$ gives
\begin{equation}
 \norm{v_X(u)}\le e^{4(d-m)bJ|\operatorname{Re}\beta|},
 \qquad \lambda_0(\beta,g)\coloneqq|\beta g|e^{4(d-m)bJ|\operatorname{Re}\beta|}.
 \label{sm:qu:support-bound}
\end{equation}
This remains valid at the boundary, where fewer interaction cells occur.
The quantity $\lambda_0$ is an absolute coefficient bound, not a
holomorphic function. Its exponential factor is a sufficient
imaginary-time estimate; no optimality is asserted.

Join distinct perturbing labels when their enlarged supports overlap. This graph has
bounded degree, depending only on $d,m,b,R,n_{\mathrm{inc}}$. A polymer $\mathcal A$ is a
nonempty connected set of distinct support labels $X$; $|\mathcal A|$
counts these labels. Its activity
$w_{\mathcal A}(\sigma;\beta,g)$ sums all Duhamel histories using exactly these
labels, with arbitrary positive insertion multiplicities. A history
specifies the ordered labels and their insertion times. For multiplicities
$n_X\ge1$, the ordered simplex and all label orderings give the
factor $\prod_{X\in\mathcal A}(n_X!)^{-1}$: there are
$n!/\prod_Xn_X!$ ordered label words and each ordered simplex has volume
$1/n!$, with $n=\sum_Xn_X$. Therefore
\begin{equation}
 |w_{\mathcal A}(\sigma;\beta,g)|
 \le\prod_{X\in\mathcal A}
       \sum_{n_X\ge1}\frac{\lambda_0^{n_X}}{n_X!}
 =(e^{\lambda_0}-1)^{|\mathcal A|}.
 \label{sm:qu:activity}
\end{equation}

Histories on disjoint enlarged supports commute and factor in the product vector
$\ket\sigma$. Summing their ordered-time shuffles therefore factors their
activities. Every history decomposes into its connected label sets, and
$\mathcal W_{\beta,g}$ is exactly the hard-core polymer partition function with
incompatibility defined by overlap of the enlarged supports. More explicitly,
\begin{equation}
 \mathcal W_{\beta,g}(\sigma)=
 \sum_{\{\mathcal A_1,\ldots,\mathcal A_r\}\,\mathrm{compatible}}
       \prod_{i=1}^r w_{\mathcal A_i}(\sigma;\beta,g).
 \label{sm:qu:polymer-identity}
\end{equation}
The sum is over unordered finite collections of polymers with pairwise disjoint
enlarged supports, including the empty collection with value one. To see the
factorization without losing time order, fix an ordering within each
component. All interleavings of these orders partition the product of
the component time simplices. Disjoint-component operators commute, and
the product-basis matrix element factors, so their summed integrals are
the product of the component integrals.

The enlarged support $U_X$ need not be connected. Add fixed local paths
to obtain a connected physical carrier $K_X$ containing $U_X$,
with uniformly bounded size and diameter. These choices can be made
translation covariant in the bulk and restricted at the boundary. For a
logarithmic cluster $\Gamma=(\mathcal A_1,\ldots,\mathcal A_r)$ set
\begin{align}
 \nu_\Gamma&\coloneqq\sum_{i=1}^r|\mathcal A_i|,\label{sm:qu:cluster-cost}\\
 Y_\Gamma&\coloneqq\bigcup_{i=1}^r\bigcup_{X\in\mathcal A_i}U_X,\label{sm:qu:cluster-support}\\
 K_\Gamma&\coloneqq\bigcup_{i=1}^r\bigcup_{X\in\mathcal A_i}K_X.\label{sm:qu:cluster-data}
\end{align}
The tuple is connected through polymer incompatibility. Define its coefficient $\phi_{\Gamma}$ as the product $\prod_{i=1}^r w_{\mathcal A_i}/r!$ multiplied by the connected coefficient $\mathcal U(\mathcal A_1,\ldots,\mathcal A_r)$ from Eq.~\eqref{sm:an:eq:cluster-series}, with hard-core Mayer factors.
Each such factor is $-1$ when the two polymers' unions of enlarged supports overlap,
including identical polymers, and zero otherwise.
Here $Y_\Gamma$ is the assigned union of enlarged cell supports from the
generating histories, whereas $K_\Gamma$ is a connected carrier containing it.
These sets remain distinct, and the assigned support is retained even
when $\phi_\Gamma$ is constant.

Polymers can repeat in a logarithmic cluster, and every occurrence
contributes to $\nu_\Gamma$. Repeated time insertions within one polymer have
already been summed in Eq.~\eqref{sm:qu:activity} and are not counted again.
There is a fixed constant $d_h$ such that
$|Y_\Gamma|\le b_h\nu_\Gamma$ and
$\diam K_\Gamma\le d_h\nu_\Gamma$.
These bounds allow $b_w=b_h$ and $d_w=d_h$ in the weighted input of Lemma~\ref{sm:lem:weighted-gas}.

To check convergence explicitly, let $D_h\ge2$ bound the number of
labels whose enlarged supports meet a given enlarged support, including the label itself.
Connected sets of $n$ labels through one specified label number at most
$D_h^{2n}$ by the rooted traversal used in the classical proof.
For any prescribed reserve $A>0$, Eq.~\eqref{sm:qu:activity} gives
\begin{equation}
 \begin{aligned}
 \sum_{\mathcal B\,\mathrm{incompatible\ with}\,\mathcal A}
       \norm{w_{\mathcal B}}_\infty e^{(1+A)|\mathcal B|}
 &\le |\mathcal A|D_h
       \sum_{n\ge1}D_h^{2n}(e^{\lambda_0}-1)^n e^{(1+A)n}\\
 &\le |\mathcal A|
 \end{aligned}
 \label{sm:qu:cluster-criterion}
\end{equation}
when $\lambda_0$ is sufficiently small. A norm here maximizes over
the cell signs on the activity's support. This is precisely the
hard-core case of Lemma~\ref{sm:lem:tree}, with cost $|\mathcal A|$
and reserve $s=A$. Its pinned bound, followed by summing the possible
root labels at one physical cell, yields
\begin{align}
 \log \mathcal W_{\beta,g}(\sigma)&=\sum_{\Gamma}\phi_{\Gamma}(\sigma;\beta,g),\label{sm:qu:local-logarithm}\\
 \sup_{\ell}\sum_{\Gamma:\ell\in K_\Gamma}
 e^{A\nu_\Gamma}\norm{\phi_{\Gamma}}_\infty
 &\le C_A\lambda_0.\label{sm:qu:cost-norm}
\end{align}
The norm maximizes the absolute value over signs in $Y_\Gamma$.

To obtain the spatial bound, mark a polymer occurrence whose carrier
contains $\ell$. This can only overcount and converts $1/r!$ to
$1/(r-1)!$, as required by the pinned bound. Hence
\begin{equation}
 \begin{aligned}
 \sum_{\Gamma:\ell\in K_\Gamma}
 e^{A\nu_\Gamma}\norm{\phi_{\Gamma}}_\infty
 &\le \sum_{\mathcal A:\ell\in\bigcup_{X\in\mathcal A}K_X}
 \norm{w_{\mathcal A}}_\infty e^{(A+1)|\mathcal A|}\\
 &\le C_h\sum_{n\ge1}D_h^{2n}
       [(e^{\lambda_0}-1)e^{A+1}]^n.
 \end{aligned}
 \label{sm:qu:spatial-root-bridge}
\end{equation}
Here $C_h$ bounds the number of microscopic carriers containing a given
physical cell; it is fixed by locality. For fixed $A$, the last geometric series
is $O(\lambda_0)$ when $\lambda_0$ is sufficiently small.

Thus every fixed exponential reserve is available by reducing the allowed
coupling. The series contains all logarithmic clusters and uses the branch
vanishing at $g=0$. Uniform absolute convergence on complex neighborhoods
makes its coefficients holomorphic.
The exponential of this logarithmic series equals the exact
$\mathcal W_{\beta,g}$, so it also proves that the ratio has no zeros
in the common convergence neighborhood. All constant terms from the
original $v_X$ are already included in this full logarithm.

Expand each local coefficient exactly in sign products:
\begin{equation}
 c_{\Gamma,Y}\coloneqq
 2^{-|Y_\Gamma|}\sum_{\sigma|_{Y_\Gamma}}
 \phi_{\Gamma}(\sigma;\beta,g)\prod_{\ell\in Y}\sigma_{\ell},
 \qquad Y\subseteq Y_\Gamma.
 \label{sm:qu:walsh}
\end{equation}
Orthogonality of these products gives
$\phi_{\Gamma}=\sum_Yc_{\Gamma,Y}\sigma_Y$ and
$\sum_Y|c_{\Gamma,Y}|\le
2^{|Y_\Gamma|}\norm{\phi_{\Gamma}}_\infty$.
Keep every constant in
\begin{equation}
 a_{0,L}\coloneqq\sum_{\Gamma}c_{\Gamma,\varnothing}.
 \label{sm:qu:scalar}
\end{equation}

For the convergence estimates only, label the nonconstant contributions
separately by $\alpha\coloneqq(\Gamma,Y)$, $Y\ne\varnothing$, and write
$k_{\alpha,L}\coloneqq c_{\Gamma,Y}$,
$\nu_\alpha\coloneqq\nu_\Gamma$,
$Y_\alpha\coloneqq Y$, and $K_\alpha\coloneqq K_\Gamma$. Equal sign products from different
clusters remain temporarily separate so that each term retains its
generating cluster cost $\nu_\alpha$. This is a proof label, not an additional
physical index; we will combine equal supports after establishing the
bounds and the bulk limit. For each fixed $S>0$ and sufficiently small
$\lambda_0$, with the threshold allowed to depend on $S$, the exact
representation and bounds are
\begin{equation}
 \boxed{
 Z_L^{\mathrm o}(\beta,g)=e^{a_{0,L}}\sum_\sigma
 \exp\left[\beta J\sum_{p\subset\Lambda_L}B_p(\sigma)
              +\sum_\alpha k_{\alpha,L}\sigma_{Y_\alpha}\right]
 }
 \label{sm:qu:partition-shadow}
\end{equation}
\begin{align}
 Q_S&\coloneqq\sup_L\sup_{\ell}\sum_{\alpha:\ell\in Y_\alpha}
              |k_{\alpha,L}|e^{S\nu_\alpha}
 \le\sup_L\sup_{\ell}\sum_{\alpha:\ell\in K_\alpha}|k_{\alpha,L}|e^{S\nu_\alpha}
 \le C_S\lambda_0,
 \label{sm:qu:classical-input}\\
 |Y_\alpha|&\le b_h\nu_\alpha,
 \qquad\diam K_\alpha\le d_h\nu_\alpha.
 \label{sm:qu:carrier-bounds}
\end{align}
Since $2^{|Y_\Gamma|}\le e^{(b_h\log2)\nu_\Gamma}$, taking
$A\coloneqq S+b_h\log2$ in Eq.~\eqref{sm:qu:cost-norm} pays the entire
Fourier coefficient count. The scalar $a_{0,L}$ is absolutely
convergent and bounded by $C_0N\lambda_0$, where $C_0$ is volume
independent: assign each cluster to any physical cell in its nonempty carrier and
use Eq.~\eqref{sm:qu:cost-norm} together with the $O(N)$ cell count. There is no bound on
an interaction's cost imposed in these identities.

For real $\beta>0$ and $g$, the grouped activities are real.
Let $\mathcal W_{\mathcal B}(\sigma)$ be the ratio in
Eq.~\eqref{sm:qu:diagonal-ratio} with $V^{\mathrm o}$ replaced by
$\sum_{X\in\mathcal B}v_X$. Retaining exactly the labels in
$\mathcal A$ is equivalent to inclusion--exclusion:
\begin{equation}
 w_{\mathcal A}(\sigma)=
 \sum_{\mathcal B\subseteq\mathcal A}
     (-1)^{|\mathcal A|-|\mathcal B|}\mathcal W_{\mathcal B}(\sigma).
 \label{sm:qu:activity-reality}
\end{equation}
Each subset Hamiltonian is Hermitian and its diagonal Boltzmann matrix
element is real, proving the claim. The logarithmic coefficients and their
Fourier coefficients are consequently real. Moreover every diagonal
matrix element of $e^{-\beta H_g^{\mathrm o}}$ is strictly positive,
and the branch of $\log \mathcal W_{\beta,g}$ continued from $g=0$ is real. This observation
identifies the physical logarithm; the absolute complex expansion is
what proves analyticity.

\subsection{Step 2: the full open-box gas and its analytic bulk limit}
\label{sm:qu:pressure}

Apply Lemma~\ref{sm:lem:weighted-gas} directly to
Eq.~\eqref{sm:qu:partition-shadow}. Its effective labels are
$\alpha=(\Gamma,Y)$, its coefficients are $k_{\alpha,L}$, and the support
and carrier bounds hold with $b_w=b_h$ and $d_w=d_h$.
The coefficients are holomorphic, and
Eq.~\eqref{sm:qu:classical-input} bounds the lemma's $Q_S$ uniformly
in volume. For sufficiently large $\operatorname{Re}(\beta J)$ and
small $\lambda_0$, the weighted row estimate
\eqref{sm:cl:weighted-row} satisfies the tree criterion
\eqref{sm:eq:tree-criterion}, with any fixed exponential cost reserve.

Denote the lemma's complete two-species collection sum for these
coefficients by $\Xi_L^{\mathrm o}(\beta,g)$. Recall that
$N_p=N_{m+1}^{\mathrm o}$ counts retained interaction cells and $2^{z_m}$
is the number of spin configurations realizing each allowed flux.
The exact normalization is
\begin{equation}
 \boxed{
 Z_L^{\mathrm o}(\beta,g)=
 e^{a_{0,L}}2^{z_m}e^{\beta J N_p}
 \left(\prod_\alpha\cosh k_{\alpha,L}\right)\Xi_L^{\mathrm o}(\beta,g).
 }
 \label{sm:qu:open-partition}
\end{equation}
The scalar, all $\cosh$ factors, and every perturbation and flux component
are included. The effective-label set is countable even in finite volume,
because logarithmic clusters may repeat. Finite-label identities extend
to this set by $\sum_\alpha|k_{\alpha,L}|<\infty$ and the lemma's uniform
absolute majorants. Thus the convergent connected series is the
holomorphic logarithm of this exact $\Xi_L^{\mathrm o}$, rather than a
formal approximation to it.

To obtain the bulk limit, we use Section~\ref{sm:bulk-limit} and verify
its local-stabilization and tail assumptions for the generated interactions.
Stabilization must be checked at fixed originating cluster cost, before
combining equal sign supports $Y$: their combined coefficient can contain
clusters of arbitrarily large extent. At fixed total cost there are only
finitely many local shapes. All time-insertion multiplicities have already
been summed in their activities, and each contribution depends only on
original terms and interaction cells within a region of diameter proportional
to that cost. Distant commuting interaction cells cancel in the conjugation, so
these coefficients agree exactly with their periodic bulk values outside
that boundary layer. The classical gas factors have the same property.

The three contributions $a_{0,L}$, $\sum_\alpha\log\cosh k_{\alpha,L}$,
and $\log\Xi_L^{\mathrm o}$ also have uniform exponential cost tails,
bounded after division by $N$ by $Ce^{-sM}$ beyond cost $M$, for some
$s>0$ on each compact subset of the strict convergence domain.
For the scalar this follows from Eq.~\eqref{sm:qu:cost-norm};
for $\log\cosh$ from Eq.~\eqref{sm:qu:classical-input} and
$|\log\cosh k|\le C|k|^2$ near zero; and for the gas from the
weighted lemma's reserved cluster bound. The fixed-cost boundary-layer
comparison and tail removal of Section~\ref{sm:bulk-limit} therefore give
\begin{equation}
 \begin{aligned}
 \psi^{\mathrm o}(\beta,g)
 &\coloneqq\lim_{L\to\infty}\frac1N\log Z_L^{\mathrm o}(\beta,g)\\
 &=\binom{d-1}{m-1}\log2+\binom d{m+1}\beta J+
     \lim_{L\to\infty}\frac1N\!\left[
       a_{0,L}+\sum_\alpha\log\cosh k_{\alpha,L}
                 +\log\Xi_L^{\mathrm o}(\beta,g)\right].
 \end{aligned}
 \label{sm:qu:open-pressure}
\end{equation}
The limit is locally uniform and holomorphic. The logarithms are the
branches fixed by their convergent expansions and combine at real
parameters to the physical logarithm in Eq.~\eqref{sm:qu:open-partition}.

Choose $B$ sufficiently large, $a$ sufficiently small, and $C\ge4(d-m)b$.
Then the quantum conditions of Theorem~\ref{sm:thm:low-temperature},
$\beta J>B$ and $\beta|g|e^{C\beta J}<a$, imply strict convergence
at the real point. Continuity gives a complex neighborhood of that point
with a volume-independent margin.

Having established these estimates and the bulk limit, we can remove the
auxiliary label $\alpha$ from the effective classical action. For each nonempty
set $Y$ of retained physical $m$-cells, combine all clusters producing the same sign
product:
\begin{equation}
 k_{Y,L}\coloneqq\sum_{\Gamma:Y\subseteq Y_\Gamma}
                         c_{\Gamma,Y}.
 \label{sm:qu:grouped-coefficient}
\end{equation}
The coefficients depend on $\beta,g$. The bound in
Eq.~\eqref{sm:qu:classical-input} makes this sum locally uniformly
absolutely convergent in the complex domain, so $k_{Y,L}$ is holomorphic.
Regrouping the full logarithm and partition function is therefore exact:
\begin{align}
 \log \mathcal W_{\beta,g}(\sigma)
 &=a_{0,L}+\sum_{Y\ne\varnothing}k_{Y,L}
                        \prod_{\ell\in Y}\sigma_\ell,
 \label{sm:qu:grouped-action}\\
 Z_L^{\mathrm o}(\beta,g)
 &=e^{a_{0,L}}\sum_\sigma\exp\left[
     \beta J\sum_{p\subset\Lambda_L}B_p(\sigma)
     +\sum_{Y\ne\varnothing}k_{Y,L}\prod_{\ell\in Y}\sigma_\ell\right].
 \label{sm:qu:grouped-partition}
\end{align}
Both sums over $Y$ run over nonempty subsets of the retained physical $m$-cells; $Y$
need not be connected.

To check spatial decay after regrouping, fix
$\mu>0$. The triangle inequality and
Eqs.~\eqref{sm:qu:classical-input} and \eqref{sm:qu:carrier-bounds} give
\begin{equation}
 \begin{aligned}
 \sup_\ell\sum_{Y\ni\ell}|k_{Y,L}|
            e^{\mu(|Y|+\diam Y)}
 &\le\sup_\ell\sum_{\alpha:\ell\in Y_\alpha}|k_{\alpha,L}|
            e^{\mu(b_h+d_h)\nu_\alpha}\\
 &\le C_{\mu(b_h+d_h)}\lambda_0.
 \end{aligned}
 \label{sm:qu:grouped-decay-bound}
\end{equation}
For each fixed $\mu$, this holds for sufficiently small $\lambda_0$.
No long-range term has been discarded. The auxiliary label $\alpha$ was needed only to retain
this generating cost in the weighted convergence and bulk-limit proof.

\subsection{Step 3: return to periodic boundaries}

For real $\beta>0$ and $g$, apply the boundary comparison
\eqref{sm:an:eq:real-comparison} to the original finite-range Hermitian
Hamiltonians. That estimate includes the $\Delta N_m$ spectator cell qubits
and requires no commutativity. It does not compare the infinite effective
interaction with a periodic one. It identifies the periodic limiting
$N^{-1}\log Z_L$ with the holomorphic open-box limit just constructed.
Thus $f(\beta,g)=-\psi^{\mathrm o}(\beta,g)/\beta$ is jointly real analytic
in the quantum domain of Theorem~\ref{sm:thm:low-temperature}. The boundary comparison
is used only at real parameters; the complex extension is supplied by
the convergent open-box expansion.

\subsection{Perturbations commuting with the unperturbed Hamiltonian}
\label{sm:qu:commuting}

Commutation of the full perturbation with $H_0$ removes the
imaginary-time norm factor from the sufficient analyticity domain.
The local terms may still be noncommuting quantum operators.

\begin{proposition}[Analyticity for a commuting total perturbation]
\label{sm:prop:commuting-analyticity}
Fix $1\le m\le d-2$, assume the locality and fixed finite spatial period
of Section~\ref{sm:cl:model}, and suppose $[H_0,V]=0$ on every sufficiently
large compatible periodic torus. There exist $B_{\mathrm c},a_{\mathrm c}>0$,
depending only on $d,m,b,R,n_{\mathrm{inc}}$, such that the full-trace free energy $f$
exists and is jointly real analytic in $(\beta,g)$ throughout
\begin{equation}
 \beta J>B_{\mathrm c},\qquad \beta|g|<a_{\mathrm c}.
 \label{sm:qu:commuting-domain}
\end{equation}
Neither $[H_0,v_X]=0$ for each original term nor $[B_p,V]=0$ for each
interaction cell is required.
\end{proposition}

\begin{proof}
We first construct an open-box interaction that commutes with
$H_0^{\mathrm o}$. This is necessary because simply deleting boundary
terms of $V$ need not preserve the commutation of the full sum.
For each original support retained in the box, define the local time average
\begin{align}
 \widehat v_X^{\mathrm o}
 &\coloneqq\frac1{2\pi}\int_0^{2\pi}
 e^{i\theta H_0^{\mathrm o}/J}v_X
 e^{-i\theta H_0^{\mathrm o}/J}\,d\theta,
 \label{sm:qu:commuting-open-projection}\\
 \widehat V^{\mathrm o}
 &\coloneqq\sum_{X\subset\Lambda_L}\widehat v_X^{\mathrm o},
 \qquad
 \widehat H_g^{\mathrm o}\coloneqq H_0^{\mathrm o}
                         +g\widehat V^{\mathrm o}.
 \label{sm:qu:commuting-open-hamiltonian}
\end{align}
The spectrum of $H_0^{\mathrm o}/J=-\sum_pB_p$ is integer valued.
In its eigenbasis, matrix elements between distinct eigenvalues acquire
nonzero integer frequencies and integrate to zero over one period.
Consequently $[H_0^{\mathrm o},\widehat v_X^{\mathrm o}]=0$.

Each averaged term is Hermitian and has norm at most one because
real-time conjugation is unitary. Only interaction cells touching $X$ enter
the conjugation; all other commuting factors cancel. Thus its support
lies in the enlarged support $U_X$ of Step~1, containing at most $b_h$ physical $m$-cells with
bounded diameter and overlap. These averages are independent of
$\beta,g$ and inherit the original finite spatial period in the bulk.

Commutation now gives
$e^{-\beta\widehat H_g^{\mathrm o}}
=e^{-\beta H_0^{\mathrm o}}e^{-\beta g\widehat V^{\mathrm o}}$.
The exact diagonal ratio is therefore
$\bra\sigma e^{-\beta g\widehat V^{\mathrm o}}\ket\sigma$.
Apply the construction of Step~1 to the projected terms, with their
supports $U_X$. Their imaginary-time conjugates are unchanged, so the
insertion estimate is one instead of $e^{4(d-m)bJ|\operatorname{Re}\beta|}$.
In particular, collecting all positive insertion multiplicities on a
connected label set $\mathcal A$ gives
\begin{equation}
 \norm{w_{\mathcal A}}_\infty
 \le\prod_{X\in\mathcal A}\sum_{n_X\ge1}
       \frac{|\beta g|^{n_X}}{n_X!}
 =(e^{|\beta g|}-1)^{|\mathcal A|}.
 \label{sm:qu:commuting-activity-bound}
\end{equation}

The polymer identity, local logarithm, and Fourier estimates therefore
hold with $\lambda_0$ replaced by $|\beta g|$ and with constants fixed
by the enlarged locality bounds. Denote the resulting scalar and
nonconstant coefficients by $\widehat a_{0,L}(\beta g)$ and
$\widehat k_{\alpha,L}(\beta g)$; the labels $\alpha$, costs
$\nu_\alpha$, supports $Y_\alpha$, and carriers $K_\alpha$ have the
same meanings as in Step~1. The exact partition function is
\begin{equation}
 \begin{aligned}
 \widehat Z_L^{\mathrm o}(\beta,g)
 &\coloneqq\operatorname{Tr}e^{-\beta\widehat H_g^{\mathrm o}}\\
 &=e^{\widehat a_{0,L}(\beta g)}\sum_\sigma
 \exp\left[\beta J\sum_{p\subset\Lambda_L}B_p(\sigma)
          +\sum_\alpha\widehat k_{\alpha,L}(\beta g)\sigma_{Y_\alpha}\right].
 \end{aligned}
 \label{sm:qu:commuting-open-partition}
\end{equation}

For each fixed reserve $S>0$ and sufficiently small $|\beta g|$,
Eqs.~\eqref{sm:qu:cost-norm}--\eqref{sm:qu:carrier-bounds} give
\begin{align}
 \sup_L\sup_\ell\sum_{\alpha:\ell\in K_\alpha}
       |\widehat k_{\alpha,L}(\beta g)|e^{S\nu_\alpha}
 &\le C_S|\beta g|,
 \label{sm:qu:commuting-weighted-bound}\\
 |\widehat a_{0,L}(\beta g)|&\le C_0N|\beta g|.
 \label{sm:qu:commuting-scalar-bound}
\end{align}
These bounds are uniform in $\beta J$ and volume. They retain all
scalar terms and the Fourier coefficient count, and provide the same
exponential cost tails. The coefficients are holomorphic in $\beta g$.
At fixed originating cost their bulk values stabilize exactly, because
the local time averages depend only on the touching interaction cells.

Lemma~\ref{sm:lem:weighted-gas} and Step~2 now apply without another
expansion. Choose $B_{\mathrm c}$ large enough for the flux bound and
$a_{\mathrm c}$ small enough for the required cost reserves and weighted
row bound. The full normalization and bulk-limit formula
\eqref{sm:qu:open-pressure}, with the hatted coefficients and including
scalar, $\log\cosh$, and gas terms, then give a locally uniform
holomorphic open-box limit on
\begin{equation}
 \operatorname{Re}(\beta J)>B_{\mathrm c},\qquad |\beta g|<a_{\mathrm c}.
 \label{sm:qu:commuting-complex-domain}
\end{equation}
The constants depend only on the enlarged locality bounds, hence on
$d,m,b,R,n_{\mathrm{inc}}$.

It remains to compare this open interaction with the original periodic
one. On the torus define
\begin{equation}
 \widehat v_X\coloneqq\frac1{2\pi}\int_0^{2\pi}
 e^{i\theta H_0/J}v_Xe^{-i\theta H_0/J}\,d\theta.
 \label{sm:qu:commuting-projection}
\end{equation}
The same integer-spectrum and locality arguments apply. Since
$[H_0,V]=0$, averaging the full sum leaves it unchanged:
\begin{equation}
 [H_0,\widehat v_X]=0,\qquad
 \norm{\widehat v_X}\le1,\qquad
 \sum_X\widehat v_X=V.
 \label{sm:qu:commuting-projection-properties}
\end{equation}

If the full enlarged support of $X$ lies in the open box, the two projected operators
agree. Only $O(L^{d-1})$ terms are omitted or changed. With $I_{\mathrm{sp}}$
the identity on the $\Delta N_m$ spectator cell qubits, this gives
\begin{align}
 \norm{V-\widehat V^{\mathrm o}\otimes I_{\mathrm{sp}}}
 &\le C_{\partial}L^{d-1},
 \label{sm:qu:commuting-boundary-interaction}\\
 \norm{H_g-\widehat H_g^{\mathrm o}\otimes I_{\mathrm{sp}}}
 &\le C_{\partial}(J+|g|)L^{d-1}.
 \label{sm:qu:commuting-boundary-norm}
\end{align}
Here $C_{\partial}$ depends only on the original locality bounds.
The projected open interaction need not equal the original restriction
$V^{\mathrm o}$; this boundary norm estimate is the needed comparison.
For real $\beta>0,g$, Eq.~\eqref{sm:an:eq:real-comparison} therefore yields
\begin{equation}
 \left|\frac{\log Z_L(\beta,g)}N
       -\frac{\log\widehat Z_L^{\mathrm o}(\beta,g)}N\right|
 \le\frac{C_{\partial}\beta(J+|g|)}{L}+\frac{\Delta N_m}{N}\log2.
 \label{sm:qu:commuting-boundary-pressure}
\end{equation}
The original periodic limit equals the holomorphic open-box limit.
Multiplying by $-1/\beta$ proves the joint real analyticity of $f$.
\end{proof}

The constants can be made common with the classical domain in
Theorem~\ref{sm:thm:low-temperature} by increasing $B$ and decreasing
$a$ as needed. This proposition gives a sufficient analyticity domain;
it does not assert the existence or location of a thermal transition.
\section{A complementary high-temperature analytic domain}
\label{sm:ht:section}

This section proves the high-temperature result, stated below as
Theorem~\ref{sm:ht:theorem}, for both classical and noncommuting
perturbations. High-temperature analyticity of general local
quantum Hamiltonians is an established result; see, for example,
Theorem~20 of the full version of
Ref.~\cite{SM-HarrowMehrabanSoleimanifar2020}.
We give the normalized-trace expansion and its convergence bound explicitly,
so that the domain used here can be checked without importing that theorem.
The only general expansion criterion we use is
Lemma~\ref{sm:lem:tree}, proved above and based on the standard connected
cluster criterion~\cite{SM-KoteckyPreiss1986}.
Unlike the low-temperature argument, this proof needs neither gauge
averaging nor a separate bound on the flux components.

The companion's high-temperature result~\cite{SM-WatanabeFullPaper2026} is
a local preparation and conversion theorem. The normalized-trace expansion
below supplies the distinct thermodynamic analyticity statement.

\subsection{Statement and local energy bound}

Use the periodic model and locality assumptions of
Sec.~\ref{sm:cl:model} throughout $0\le m\le d-1$, with $N=L^d$
unit cells and $n\coloneqq\binom dmN$ physical $m$-cell qubits.
Retain scalar terms on their assigned nonempty supports, so their
normalization contribution is included throughout.
Combine all Hamiltonian terms with the same assigned cell support $Y$:
\begin{equation}
 H_g=\sum_Y h_Y(g),\qquad
 h_Y(g)\coloneqq -J\sum_{p:\partial p=Y}B_p+g v_Y .
 \label{sm:ht:hamiltonian}
\end{equation}
Here $v_Y$ is the perturbing term assigned to $Y$; it is zero if no such
term is present, and coincident terms, if any, are added. The sum over
$Y$ includes every support occurring in either $H_0$ or $V$.
Triangle inequalities and the  $2(d-m)$ interaction cells incident on each physical cell give
\begin{equation}
 b'\coloneqq\max\{2(m+1),b\},\qquad
 |Y|\le b',\qquad
 \sup_{\ell}\sum_{Y:\ell\in Y}\|h_Y(g)\|
 \le 2(d-m)J+n_{\mathrm{inc}}|g|.
 \label{sm:ht:energy-bound}
\end{equation}
The supremum is over physical cells $\ell$, and the norm is the operator norm.
The left side measures interaction strength incident on one physical cell, rather
than the extensive norm of the whole Hamiltonian.

\begin{theorem}[High-temperature free energy]
\label{sm:ht:theorem}
For every fixed $d,m$ with $0\le m\le d-1$, the explicit constant
\begin{equation}
 a'\coloneqq\frac{1}{2e^4b'}
 \label{sm:ht:constant}
\end{equation}
depends only on $m,b$ and is sufficient for the following statement.
For every allowed $V$, the periodic logarithm of
the partition function per cell and the free-energy density,
\begin{equation}
 \psi(\beta,g)\coloneqq\lim_{L\to\infty}\frac{\log Z_L(\beta,g)}{N},
 \qquad f(\beta,g)\coloneqq-\frac{\psi(\beta,g)}{\beta},
 \label{sm:ht:pressure}
\end{equation}
exist and are jointly real analytic in the open domain
\begin{equation}
 \beta[2(d-m)J+n_{\mathrm{inc}}|g|]<a'.
 \label{sm:ht:domain}
\end{equation}
No mutual commutativity, sign condition, local gauge or wrapping
$m$-form symmetry, or separate smallness assumption on $|g|/J$ is required.
\end{theorem}

The number $e$ in Eq.~\eqref{sm:ht:constant} is the base of the natural
logarithm. This bound is deliberately conservative. At any fixed finite
$g$, it applies for sufficiently large temperature; it is not an estimate
of the actual phase boundary.

\subsection{Step 1: an exact trace expansion that keeps operator ordering}

Allow $\beta,g$ to be complex and define
\begin{align}
 D_Y(\beta,g)&\coloneqq-\beta h_Y(g),
 \label{sm:ht:complex-terms}\\
 \tau_L(O)&\coloneqq2^{-n}\Tr O,\qquad
 \widehat Z_L(\beta,g)\coloneqq
 \tau_L\!\left(e^{\sum_Y D_Y(\beta,g)}\right).
 \label{sm:ht:normalized-trace}
\end{align}
The normalized trace $\tau_L$ is a product over physical $m$-cell qubits and
factorizes for operators on disjoint supports. Exactly,
$Z_L(\beta,g)=2^n\widehat Z_L(\beta,g)$.
Expanding the single operator exponential gives
\begin{equation}
 \widehat Z_L(\beta,g)=
 \sum_{r\ge0}\frac{1}{r!}
 \sum_{Y_1,\ldots,Y_r}
 \tau_L(D_{Y_1}\cdots D_{Y_r}).
 \label{sm:ht:word-series}
\end{equation}
The $r=0$ term is one. A word $(Y_1,\ldots,Y_r)$ records the
order of multiplication; repetitions are allowed. At finite volume the
series is absolutely convergent, since its absolute sum is bounded by
$\exp(\sum_Y\|D_Y\|)$. Regrouping its terms is therefore legitimate.
There has been no replacement of an exponential of a sum by a product
of exponentials.

Regard each support $Y$ as a label, and join two labels when those
supports overlap. A \emph{polymer}
$P$ is a nonempty connected set of distinct labels in this overlap graph.
Its cost $|P|$ counts labels, and its support is
$Y_P\coloneqq\bigcup_{Y\in P}Y$.
Define its activity by summing every word whose set of used labels is
exactly $P$:
\begin{equation}
 \omega_P(\beta,g)\coloneqq
 \sum_{r\ge|P|}\frac{1}{r!}
 \sum_{\substack{Y_1,\ldots,Y_r\in P\\
                 \{Y_1,\ldots,Y_r\}=P}}
 \tau_L(D_{Y_1}\cdots D_{Y_r}).
 \label{sm:ht:activity}
\end{equation}

Two polymers are incompatible if their supports overlap; a polymer is
incompatible with itself. The exact expansion becomes
\begin{equation}
 \boxed{
 \widehat Z_L(\beta,g)=
 \sum_{\{P_1,\ldots,P_r\}\,\mathrm{compatible}}
       \prod_{i=1}^r\omega_{P_i}(\beta,g).
 }
 \label{sm:ht:polymer-sum}
\end{equation}
The sum is over unordered compatible collections of any size, including
the empty collection. To verify the equality, separate a word into the
connected components of its used labels. Different components act on
disjoint tensor factors, so their operators commute with one another and
their normalized traces factor. If component words have lengths
$r_1,\ldots,r_s$ and total length $r$, their number of interleavings is
$r!/\prod_i r_i!$. This cancels the $1/r!$ in
Eq.~\eqref{sm:ht:word-series} and leaves exactly the product of the
factors $1/r_i!$ in Eq.~\eqref{sm:ht:activity}. Ordering \emph{within} each
component is preserved. This is why the representation also holds for
arbitrary noncommuting local terms.

\subsection{Step 2: sum repeated insertions and control connected sets}

Introduce nonnegative bounds, not additional complex couplings,
\begin{equation}
 q_Y\coloneqq\|D_Y(\beta,g)\|,\qquad
 \zeta_Y\coloneqq e^{q_Y}-1,\qquad
 \kappa\coloneqq|\beta|(2(d-m)J+n_{\mathrm{inc}}|g|).
 \label{sm:ht:majorants}
\end{equation}
The local triangle inequality gives
$\sup_{\ell}\sum_{Y:\ell\in Y}q_Y\le\kappa$, also for complex
$\beta,g$. Since every $Y$ is nonempty, $q_Y\le\kappa$ as well.
The bounds involve absolute values and need not be holomorphic;
the operators $D_Y$ are holomorphic in both parameters.

The trace of a word has absolute value at most the product of its
operator norms. For fixed multiplicities $r_Y\ge1$ and
$r=\sum_{Y\in P}r_Y$, there are $r!/\prod_{Y\in P}r_Y!$ ordered
words. Their common norm bound is $\prod_{Y\in P}q_Y^{r_Y}$.
The factor $1/r!$ thus leaves $\prod_{Y\in P}q_Y^{r_Y}/r_Y!$.
Summing the multiplicities gives
\begin{equation}
 |\omega_P|
 \le\prod_{Y\in P}
       \left(\sum_{r_Y\ge1}\frac{q_Y^{r_Y}}{r_Y!}\right)
 =\prod_{Y\in P}\zeta_Y.
 \label{sm:ht:activity-bound}
\end{equation}
Thus repetitions are already controlled before the connected-cluster
expansion is taken. Also $\zeta_Y\le e^\kappa q_Y$.

Write $Y\sim Y'$ when $Y\cap Y'\ne\varnothing$,
including $Y\sim Y$. Summing first over the physical cells of $Y$
gives
\begin{equation}
 \sum_{Y'\sim Y}\zeta_{Y'}
 \le b'e^\kappa\kappa.
 \label{sm:ht:neighbor-bound}
\end{equation}
Counting a neighbor more than once only enlarges this bound.

We need the corresponding sum over connected sets, not just neighboring
labels. Weight each label by $e^2\zeta_Y$ to leave an exponential
reserve for the logarithm below. Root a connected label set
at a specified label $Y$. Every such set has a spanning tree.
Allowing repeated labels in tree branches only increases its weighted
sum. A convenient upper bound is the increasing tree recursion
\begin{equation}
 T_Y^{(0)}\coloneqq e^2\zeta_Y,\qquad
 T_Y^{(d_{\rm t}+1)}\coloneqq
 e^2\zeta_Y\exp\!\left(\sum_{Y'\sim Y}T_{Y'}^{(d_{\rm t})}\right),
 \qquad d_{\rm t}=0,1,\ldots .
 \label{sm:ht:tree-recursion}
\end{equation}
Here $d_{\rm t}$ counts allowed tree depth. The exponential sums over any number
of children, with their factorials. For distinct child labels those
factorials cancel their orderings, so every spanning tree of a connected
set is included with its full product weight.

If $\kappa<a'$, then $\kappa<1$ and
\begin{equation}
 \sum_{Y'\sim Y}e^3\zeta_{Y'}
 \le b'e^{\kappa+3}\kappa<\frac12.
 \label{sm:ht:tree-smallness}
\end{equation}
If $T_{Y'}^{(d_{\rm t})}\le e^3\zeta_{Y'}$ for every $Y'$, the recursion
gives $T_Y^{(d_{\rm t}+1)}\le e^{5/2}\zeta_Y\le e^3\zeta_Y$.
The bound holds at depth zero and therefore at every depth by induction.
Letting the depth grow bounds all connected sets containing the root:
\begin{equation}
 \sum_{P\ni Y}e^{2|P|}\prod_{Y'\in P}\zeta_{Y'}
 \le e^3\zeta_Y.
 \label{sm:ht:rooted-set-bound}
\end{equation}
This establishes the required connected-set bound directly from the
local interaction strength, without a bound on $\|H_g\|$.

\subsection{Step 3: convergence of the logarithm}

If $P'$ is incompatible with $P$, some label of $P'$ overlaps a label
of $P$. Equations~\eqref{sm:ht:activity-bound},
\eqref{sm:ht:rooted-set-bound}, and \eqref{sm:ht:tree-smallness} imply
\begin{equation}
 \begin{aligned}
 \sum_{P'\,\mathrm{incompatible\ with}\,P}
       |\omega_{P'}|e^{2|P'|}
 &\le\sum_{Y\in P}\sum_{Y'\sim Y}
       \sum_{P'\ni Y'}e^{2|P'|}\prod_{Y''\in P'}\zeta_{Y''}\\
 &\le\sum_{Y\in P}\sum_{Y'\sim Y}e^3\zeta_{Y'}
 <\frac{|P|}{2}.
 \end{aligned}
 \label{sm:ht:cluster-criterion}
\end{equation}
This verifies Lemma~\ref{sm:lem:tree} with cost $c(P)\coloneqq|P|$, reserve
$s=1$, and pair factor zero for incompatible polymers and one otherwise.
For clarity, its conclusion here is the absolutely convergent series
\begin{equation}
 \log\widehat Z_L(\beta,g)=
 \sum_{r\ge1}\frac{1}{r!}\sum_{P_1,\ldots,P_r}
       \mathcal U(P_1,\ldots,P_r)\prod_{i=1}^r\omega_{P_i}(\beta,g).
 \label{sm:ht:log-series}
\end{equation}
The coefficient $\mathcal U$ is the sum over connected graphs on the
$r$ positions: each graph edge contributes $-1$ when its two polymers
are incompatible and zero otherwise. For $r=1$ the coefficient is one.
The tree identity and rooted recursion in the proof of
Lemma~\ref{sm:lem:tree} bound this series even with its alternating signs.

Each activity is holomorphic in $(\beta,g)$, as follows from its uniformly
convergent word series on compact parameter sets. The bounds above are
locally uniform in the complex domain
\begin{equation}
 \mathcal D_{\rm HT}\coloneqq
 \{(\beta,g)\in\mathbb C^2:|\beta|(2(d-m)J+n_{\mathrm{inc}}|g|)<a'\}.
 \label{sm:ht:complex-domain}
\end{equation}
Therefore Eq.~\eqref{sm:ht:log-series} is a holomorphic logarithm, with
value zero at $\beta=0$ for every $g$. Its exponential is the exact
$\widehat Z_L$, so the finite-volume partition function has no zeros
in this domain. At real physical parameters this branch agrees with
the logarithm of the positive partition function.

\subsection{Step 4: periodic bulk limit and thermodynamic meaning}

It remains to justify passage to infinite volume; finite-volume
analyticity alone would not suffice. A logarithmic cluster is a tuple
$\mathcal C=(P_1,\ldots,P_r)$ contributing to
Eq.~\eqref{sm:ht:log-series}. Define its total cost and contribution by
\begin{equation}
 \operatorname{cost}(\mathcal C)\coloneqq\sum_{i=1}^r|P_i|,\qquad
 \phi(\mathcal C)\coloneqq\frac{\mathcal U(P_1,\ldots,P_r)}{r!}
                          \prod_{i=1}^r\omega_{P_i}.
 \label{sm:ht:cluster-data}
\end{equation}
We verify the inputs to the bulk-limit argument of
Sec.~\ref{sm:bulk-limit}. Connected support overlaps bound each
cluster's diameter by a constant times $\operatorname{cost}(\mathcal C)$.
At any fixed cost cutoff and sufficiently large $L$, the clusters are
nonwrapping and have finitely many relative support patterns per spatial
period. Their activities, including all repeated insertions, coincide
with those of the corresponding fixed local patterns: normalized traces
of spectator qubits equal one. Thus the local weights stabilize, and
averaging their translations over one period gives the fixed-cost limit.

There is a uniform bound on the remaining large-cost clusters. The
rooted version of Lemma~\ref{sm:lem:tree} and the same estimates give
\begin{align}
 \sup_{\ell}\sum_{P:\ell\in Y_P}|\omega_P|e^{2|P|}
 &\le e^{\kappa+3}\kappa,
 \label{sm:ht:cell-root-bound}\\
 \sum_{\mathcal C}|\phi(\mathcal C)|e^{\operatorname{cost}(\mathcal C)}
 &\le\binom dmN e^{\kappa+3}\kappa,
 \label{sm:ht:reserved-total}\\
 \frac1N\sum_{\mathcal C:\operatorname{cost}(\mathcal C)\ge M}|\phi(\mathcal C)|
 &\le \binom dm e^{\kappa+3}\kappa\,e^{-M},\qquad M\ge1.
 \label{sm:ht:bulk-tail}
\end{align}
For the first inequality, root a polymer at any constituent label whose
support contains $\ell$ and use Eq.~\eqref{sm:ht:rooted-set-bound}.
For the second, use the pinned bound with reserve $s=1$ and sum the
roots, as in Eq.~\eqref{sm:cl:reserved-total}. Assigning each root
polymer to a physical cell and summing the $\binom dmN$ physical cells can only overcount;
the first inequality then bounds this root sum.
On $\operatorname{cost}(\mathcal C)\ge M$, the extra weight is at least $e^M$, which
gives the third inequality.

The fixed-cost limit and uniform tail satisfy the remaining hypotheses
of the argument in Sec.~\ref{sm:bulk-limit}. Hence
\begin{equation}
 \boxed{
 \psi(\beta,g)=\binom dm\log2+
       \lim_{L\to\infty}\frac1N\log\widehat Z_L(\beta,g)
 }
 \label{sm:ht:analytic-pressure}
\end{equation}
exists locally uniformly on $\mathcal D_{\rm HT}$ and is holomorphic
there. The term $\binom dm\log2$ restores the full physical cell-Hilbert-space
normalization; scalar weights and all sectors have been retained.
For real $\beta>0$, $f=-\psi/\beta$ is jointly real analytic in
Eq.~\eqref{sm:ht:domain}, proving Theorem~\ref{sm:ht:theorem}.
The dimensionless quantity $-\beta f=\psi$ also extends analytically to
$\beta=0$ at fixed $g$, with value $\binom dm\log2$. The free-energy density
$f$ itself has the usual entropy term $-\binom dm\log2/\beta$ there; no
analyticity of $f$ at infinite temperature is asserted.
At any fixed $g$, decreasing $\beta$ within this domain gives an analytic
continuation of $\psi$ to the infinite-temperature point $\beta=0$.

The high- and low-temperature
analyticity proofs control two interiors by different expansions.
Neither convergence bound determines what occurs between them or, by
itself, proves that a thermodynamic singularity separates them.
\section{Extension to independent onsite ancillary systems}
\label{sm:anc:section}

The analyticity results survive the addition of independent finite-dimensional
onsite systems and sufficiently weak local couplings to them. These ancillary
systems are physical degrees of freedom included in the partition function;
they are not disposable registers used to implement a channel. We retain their
full free-energy contribution. The corresponding state-conversion problem is
treated separately in Ref.~\cite{SM-WatanabeFullPaper2026}.
Integration over independent onsite degrees of freedom and its local
expansion are shared inputs. The complex-reference bounds and ancillary
free-energy normalization below provide the additional control needed
for analyticity.

\subsection{Joint model and analyticity domains}

Add a bounded number of ancillary sites per unit cell, with uniformly bounded
Hilbert-space dimensions. Label them by $j$, and let $h_j$ be their onsite
Hermitian Hamiltonians, independent of $\beta,g,L$. Define
\begin{align}
 H_A&\coloneqq\sum_j h_j,\qquad
 \lambda_{\max}(h_j)-\lambda_{\min}(h_j)\le\Delta_A,
 \label{sm:anc:onsite}\\
 H_{g,A}&\coloneqq H_0+H_A+gV,\qquad V\coloneqq\sum_Xv_X,
 \label{sm:anc:joint-hamiltonian}\\
 Z_{\mathrm{joint},L}(\beta,g)&\coloneqq\Tr e^{-\beta H_{g,A}},
 \label{sm:anc:joint-partition}\\
 f_{\mathrm{joint}}(\beta,g)&\coloneqq
 -\lim_{L\to\infty}\frac{\log Z_{\mathrm{joint},L}(\beta,g)}{\beta N}.
 \label{sm:anc:free-energy}
\end{align}
Here $\lambda_{\max}$ and $\lambda_{\min}$ denote the largest and smallest
eigenvalues, and $\Delta_A$ is a fixed nonnegative energy bound.

Each physical qubit and each ancillary factor counts as one site. Assigned
supports $X$ are nonempty, including for scalar terms, with $|X|\le b$,
$\|v_X\|\le1$, $\diam X\le R$, and at most $n_{\mathrm{inc}}$ perturbing supports contain
any site of either kind. Thus $V$ can include system--ancilla, ancilla--ancilla,
and ancilla-only terms. Assign ancillary sites to nearby unit cells when
measuring distances. Their placement, dimensions, onsite Hamiltonians, and
the perturbation have a common fixed finite spatial period. All these data
are fixed before the thermodynamic limit along compatible sizes.

Call the joint perturbation classical when it is diagonal in the system
$Z$ basis and a fixed eigenbasis of every $h_j$. In this case replace each
$v_X$ by its diagonal part in that basis; their sum remains $V$ and the
norm and support bounds do not increase. Quantum perturbations are
arbitrary local Hermitian operators on the joint system.

\begin{proposition}[Analyticity with physical ancillas]
\label{sm:anc:proposition}
For $1\le m\le d-2$, there are constants $B_A,a_A,C_A>0$ such that
$f_{\mathrm{joint}}$ exists and is jointly real analytic in $(\beta,g)$ in
the respective open domains
\begin{align}
 \beta J>B_A,\qquad\beta|g|<a_A
 &\quad\text{(classical)},\label{sm:anc:classical-domain}\\
 \beta J>B_A,\qquad\beta|g|e^{C_A\beta(J+\Delta_A)}<a_A
 &\quad\text{(quantum)}.\label{sm:anc:quantum-domain}
\end{align}
The constants depend only on the fixed joint locality and geometric data,
including the bounds on ancillary number and dimension, not on onsite Gibbs
probabilities or scalar energy offsets. For $0\le m\le d-1$, the same joint
free energy is analytic in the high-temperature domain
\begin{equation}
 \beta[2(d-m)J+\Delta_A/2+n_{\mathrm{inc}}|g|]<a'_A,
 \qquad a'_A\coloneqq\frac1{2e^4\max\{2(m+1),b\}}.
 \label{sm:anc:high-domain}
\end{equation}
\end{proposition}
No onsite spectral gap or nondegenerate ground state is required. Independence
of the reference onsite systems is essential to the proof; an arbitrary
interacting ancillary reference is not covered.

\subsection{Exact elimination of the ancillary degrees of freedom}

For low temperature use the free box of Sec.~\ref{sm:cl:open-box}. Keep the
onsite ancillary terms, and delete perturbing terms crossing the boundary;
the superscript $\mathrm o$ denotes these restrictions. Put
\begin{align}
 Z_j^A(\beta)&\coloneqq\Tr_j e^{-\beta h_j},\qquad
 \tau_j(\beta)\coloneqq Z_j^A(\beta)^{-1}e^{-\beta h_j},
 \label{sm:anc:local-reference}\\
 Z_{A,L}(\beta)&\coloneqq\prod_jZ_j^A(\beta),\qquad
 \tau_A(\beta)\coloneqq\bigotimes_j\tau_j(\beta).
 \label{sm:anc:product-reference}
\end{align}
$\Tr_j$ is the trace on ancillary site $j$ and $\Tr_A$ the trace on all
ancillary sites.

For a system spin configuration $\sigma$, define
\begin{equation}
 \mathcal W^A_{\beta,g}(\sigma)\coloneqq
 \frac{e^{\beta H_0^{\mathrm o}(\sigma)}}{Z_{A,L}(\beta)}
 \Tr_A\bra\sigma e^{-\beta H_{g,A}^{\mathrm o}}\ket\sigma.
 \label{sm:anc:ratio}
\end{equation}
Taking the remaining system trace gives the exact identity
\begin{equation}
 \boxed{Z_{\mathrm{joint},L}^{\mathrm o}
 =Z_{A,L}\sum_\sigma e^{-\beta H_0^{\mathrm o}(\sigma)}
                          \mathcal W^A_{\beta,g}(\sigma).}
 \label{sm:anc:exact-reduction}
\end{equation}
Thus only the local logarithm of $\mathcal W^A$ must be controlled to reuse
the gauge expansion. The factor $Z_{A,L}$ must not be discarded.
At real $\beta>0,g$, positivity of the Boltzmann operator gives
$\mathcal W^A_{\beta,g}(\sigma)>0$, and at $g=0$ the ratio equals one.
Consequently the convergent logarithm continued from $g=0$ agrees with
the physical real logarithm.

For classical couplings, let $z$ denote the ancillary eigenbasis labels and
write $V^{\mathrm o}(\sigma,z)$ for the diagonal value. Then
\begin{equation}
 \mathcal W^A_{\beta,g}(\sigma)
 =\mathbb E_{z\sim\tau_A(\beta)}e^{-\beta gV^{\mathrm o}(\sigma,z)}.
 \label{sm:anc:classical-weight}
\end{equation}
The effective system interaction is $-\beta^{-1}\log\mathcal W^A$, not the
average of $gV$.

Expand the product of factors
$1+[e^{-\beta gv_X(\sigma,z)}-1]$ and group selected supports by connected
overlap. For a connected set $\mathcal A$ of distinct labels, its activity
is the expectation of their product under the product state $\tau_A(\beta)$. At real $\beta$, product
expectation contracts their supremum norms, giving
\begin{equation}
 \|w_{\mathcal A}\|_\infty
 \le(e^{|\beta g|}-1)^{|\mathcal A|}.
 \label{sm:anc:classical-activity}
\end{equation}
Disjoint components factor because the reference ancillary state is a
product. The bounded overlap degree therefore gives the connected logarithm
for small $|\beta g|$, regardless of the onsite bias.

For quantum couplings, the Duhamel derivation in
Sec.~\ref{sm:qu:shadow} applies with reference $H_0^{\mathrm o}+H_A$:
\begin{align}
 v_X^A(u)&\coloneqq
 e^{u\beta(H_0^{\mathrm o}+H_A)}v_Xe^{-u\beta(H_0^{\mathrm o}+H_A)},
 \qquad0\le u\le1,\label{sm:anc:imaginary-time}\\
 \mathcal W^A_{\beta,g}(\sigma)
 &=\Tr_A\!\left[\tau_A(\beta)\bra\sigma
 \mathcal T\exp\!\left(-\beta g\int_0^1du\sum_Xv_X^A(u)\right)
 \ket\sigma\right].\label{sm:anc:duhamel}
\end{align}

Only gauge interactions touching $X$ and onsite terms on ancillary sites in
$X$ enter the conjugation. The support lies in the same bounded enlarged support on the gauge
spins, together with those ancillary sites. Spectral widths give
\begin{align}
 \|v_X^A(u)\|&\le
 e^{[4b(d-m)J+b\Delta_A]|\operatorname{Re}\beta|},
 \label{sm:anc:conjugation-bound}\\
 \lambda_A(\beta,g)&\coloneqq|\beta g|
 e^{[4b(d-m)J+b\Delta_A]|\operatorname{Re}\beta|}.
 \label{sm:anc:quantum-strength}
\end{align}
Each ancillary conjugation costs at most $e^{\Delta_A|\operatorname{Re}\beta|}$;
scalar energy shifts cancel from it. Grouping histories by their distinct
perturbation labels and summing all positive insertion multiplicities, as in
Eq.~\eqref{sm:qu:activity}, bounds their activities by
$(e^{\lambda_A}-1)^{|\mathcal A|}$ at real $\beta$.

\subsection{Complex convergence, normalization, and the bulk limit}

Real bounds alone would not prove analyticity in temperature, because the
reference $\tau_A(\beta)$ also varies with $\beta$. Fix a real $\beta_0>0$.
For the following estimate only, subtract the spectral midpoint from each
$h_j$, obtaining $\bar h_j$ with $\|\bar h_j\|\le\Delta_A/2$. Normalized
onsite states are unchanged, and
\begin{equation}
 \tau_j(\beta)=
 \frac{\tau_j(\beta_0)e^{-(\beta-\beta_0)\bar h_j}}
 {\Tr_j[\tau_j(\beta_0)e^{-(\beta-\beta_0)\bar h_j}]}.
 \label{sm:anc:complex-reference}
\end{equation}
If $|\beta-\beta_0|\Delta_A/2<\log(3/2)$, the denominator differs from one
by less than $1/2$, while the numerator has trace norm less than $3/2$.
Consequently $Z_j^A(\beta)\ne0$ and $\|\tau_j(\beta)\|_1<3$ uniformly in $j$.
For $\Delta_A=0$ the normalized states are constant.

A component with
$|\mathcal A|$ labels touches at most $b|\mathcal A|$ ancillary sites,
so the complex versions of the activity bounds are
\begin{equation}
 \|w_{\mathcal A}\|_\infty\le
 \begin{cases}
 [3^b(e^{|\beta g|}-1)]^{|\mathcal A|},&\text{classical},\\
 [3^b(e^{\lambda_A}-1)]^{|\mathcal A|},&\text{quantum}.
 \end{cases}
 \label{sm:anc:complex-activities}
\end{equation}
The fixed factor $3^b$ is absorbed in the smallness threshold. No bound on
the trace norm of the full-volume complex $\tau_A$ is used. Thus all local
activities and their connected logarithm are holomorphic on a neighborhood
of each allowed real point, uniformly in volume.

Expand each local logarithmic contribution in system sign products, exactly
as in Eqs.~\eqref{sm:qu:walsh}--\eqref{sm:qu:carrier-bounds}. Retain its
original generating cost $\nu_\alpha$ and carrier, including the integrated
ancillary sites and fixed short paths to nearby physical cells. This keeps
the diameter proportional to cost even if the remaining sign support
$Y_\alpha$ is small. The bounded number of ancillary sites per cell permits
the same root counting. With the coefficient labels of that construction,
\begin{align}
 \log\mathcal W^A_{\beta,g}(\sigma)
 &=a_{0,L}+\sum_\alpha k_{\alpha,L}\sigma_{Y_\alpha},
 \label{sm:anc:effective-log}\\
 \sup_L\sup_\ell\sum_{\alpha:\ell\in K_\alpha}
       |k_{\alpha,L}|e^{S\nu_\alpha}
 &\le\begin{cases}C_S|\beta g|,&\text{classical},\\
                   C_S\lambda_A,&\text{quantum}.
       \end{cases}\label{sm:anc:weighted-bound}
\end{align}
Here $\ell$ is a physical cell, $K_\alpha$ is the retained generating
carrier, and $S$ is any prescribed finite exponential reserve, with a
correspondingly small coupling threshold. Ancilla-only clusters and all
constant Fourier terms belong to $a_{0,L}$; its absolute sum per unit cell
has the same smallness and exponential cost-tail bounds.

Lemma~\ref{sm:lem:weighted-gas} now applies without changing the gauge or
defect geometry. In the notation of that lemma, the full normalization is
\begin{equation}
 \boxed{Z_{\mathrm{joint},L}^{\mathrm o}
 =Z_{A,L}e^{a_{0,L}}2^{z_m}e^{\beta JN_p}
       \left(\prod_\alpha\cosh k_{\alpha,L}\right)\Xi_L^{\mathrm o}.}
 \label{sm:anc:normalized-gas}
\end{equation}

The fixed spatial period makes each local coefficient stabilize in the
bulk. The reserved cost tails justify the locally uniform limit exactly
as in Sec.~\ref{sm:qu:pressure}. In addition,
$N^{-1}\log Z_{A,L}=N^{-1}\sum_j\log Z_j^A$ tends to its period average,
which is holomorphic near $\beta_0$. Thus onsite scalar energies and the
ancillary entropy remain in the full free energy.

For the periodic comparison, use the same ancillary sites and onsite
Hamiltonians in both boxes; only boundary-crossing couplings are removed.
Embed any missing gauge spins as uncoupled spectator qubits as in
Sec.~\ref{sm:cl:boundary-comparison}. The Hamiltonian difference has norm
$O((J+|g|)L^{d-1})$, and the spectator normalization is also of surface order.
The real trace comparison \eqref{sm:an:eq:real-comparison} identifies the
periodic free energy with the analytic open-box limit. Choosing $B_A$ large,
$a_A$ small, and $C_A\ge\max\{4b(d-m),b\}$ proves the two low-temperature
domains. These choices do not require a lower bound on an onsite Gibbs
probability.

\subsection{High-temperature extension}

Center each $h_j$ as above for the expansion, retaining its removed scalar
in the final partition function. For the centered joint Hamiltonian, the
local energy bound on any site $x$, physical or ancillary, is
\begin{equation}
 \sup_x\sum_{Y\ni x}\|\bar h_Y(g)\|
 \le2(d-m)J+\Delta_A/2+n_{\mathrm{inc}}|g|.
 \label{sm:anc:high-energy}
\end{equation}
Here $\bar h_Y(g)$ are its terms grouped by their joint supports; their
size is at most $\max\{2(m+1),b\}$. The normalized trace is the full trace
divided by the joint Hilbert-space dimension. It factors on disjoint
supports and has norm one, regardless of the onsite dimensions. All word,
tree, and logarithmic estimates in Sec.~\ref{sm:ht:section} therefore apply
with the right side of Eq.~\eqref{sm:anc:high-energy} as local energy scale,
giving precisely $a'_A$ in Eq.~\eqref{sm:anc:high-domain}.

To track normalization explicitly, let $d_{A,j}$ be the dimension of
ancillary site $j$ and $c_j\coloneqq[\lambda_{\max}(h_j)+
\lambda_{\min}(h_j)]/2$ its subtracted scalar. If $\widehat Z_{\mathrm{joint},L}$ is the
normalized partition function of the centered joint Hamiltonian, then
\begin{equation}
 Z_{\mathrm{joint},L}
 =e^{-\beta\sum_jc_j}\,2^{\binom dmN}
       \left(\prod_jd_{A,j}\right)\widehat Z_{\mathrm{joint},L}.
 \label{sm:anc:high-normalization}
\end{equation}

There are still only $O(N)$ joint sites, so the extensive root sum and
exponential cost tail in Sec.~\ref{sm:ht:section} remain bounded per unit
cell. Finite periodicity gives the limit of every term in
\begin{equation}
 \frac{\log Z_{\mathrm{joint},L}}N
 =\binom dm\log2+\frac1N\sum_j\log d_{A,j}
       -\frac\beta N\sum_jc_j+\frac{\log\widehat Z_{\mathrm{joint},L}}N.
 \label{sm:anc:high-free-energy}
\end{equation}
The limit is holomorphic in the complex high-temperature domain and yields
joint real analyticity of $f_{\mathrm{joint}}$ at $\beta>0$. As before,
$-\beta f_{\mathrm{joint}}$ extends to $\beta=0$, while the free energy itself
has the entropy term fixed by the full joint Hilbert space. This completes
the proof of Proposition~\ref{sm:anc:proposition}.


\begingroup
\makeatletter
\let\@FMN@list\@empty
\let\combined@label\label
\def\label#1{\combined@label{SM-bib-#1}}
\makeatother

\endgroup

\begin{thebibliography}{16}
\bibitem [{Note1()}]{SM-Note1}%
  \BibitemOpen
  \bibinfo {note} {For example, $N^{-1}\log [2\cosh (Nx)]$ converges uniformly
  on the real axis to the nonanalytic function $|x|$; the nearest zeros of
  $2\cosh (Nx)$ are $\pm i\pi /(2N)$.}\BibitemShut {Stop}%
\end{thebibliography}
\end{document}